%% file: 00_main.tex
\documentclass[runningheads,11pt]{llncs}
 
\usepackage[margin=1in]{geometry}
\usepackage{graphicx}
\usepackage{algorithm}
\usepackage{algpseudocode}
\usepackage{booktabs}
\usepackage{bbm}
\usepackage{tikz}
\usepackage{mycommands}
\usepackage{framed}
\usepackage[textsize=footnotesize]{todonotes}
\usepackage{amsmath,amsfonts,amssymb,amsthm}
\usepackage{tikz}
\usepackage{tcolorbox}
\usepackage{graphicx}
\usepackage{subcaption}
\usepackage{cite}
\usepackage{xspace}
\usepackage{svg}
\usepackage{dsfont}
\usepackage{mathtools}
\usepackage{hyperref}
\hypersetup{
    colorlinks=true,
    linkcolor=blue,
    filecolor=magenta,
    urlcolor=red,
    citecolor=red
}
\usepackage[capitalise,nameinlink]{cleveref}
\usepackage{thmtools, thm-restate}
\usepackage[T1]{fontenc}
\usepackage[utf8]{inputenc}
\usepackage{lmodern}
\usepackage{float}
\usepackage{enumerate}

\title{Computing Diverse and Nice Triangulations}

\author{Waldo Gálvez\inst{1} \and
Mayank Goswami\inst{2} \and
Arturo Merino\inst{3} \and \\
GiBeom Park\inst{4,5} \and Meng-Tsung Tsai\inst{6}}

\authorrunning{W. Gálvez, M. Goswami, A. Merino, G. Park, and M. Tsai}

\institute{Universidad de Concepción, Chile \and Queens College of CUNY, USA \and Universidad de O'Higgins, Chile \and CUNY Graduate Center, USA \and Academia Sinica, Taiwan}

\begin{document}

\makeatletter
  \renewcommand*\l@title[2]{}
  \renewcommand*\l@author[2]{}
  \makeatother

\maketitle
\footnotetext[5]{A preliminary version of this paper was presented at the 31st Annual Fall Workshop on Computational Geometry 2024.}

\begin{abstract}
\noindent We initiate the study of computing diverse triangulations to a given polygon.
Given a simple $n$-gon $P$, an integer $ k \geq 2 $, a quality measure $\sigma$ on the set of triangulations of $P$ and a factor $ \alpha \geq 1 $, we formulate the Diverse and Nice Triangulations (DNT) problem that asks to compute $k$ \emph{distinct} triangulations $T_1,\dots,T_k$ of $P$ such that a) their diversity, $\sum_{i < j} d(T_i,T_j)
$, is as large as possible \emph{and} b) they are nice, i.e., $\sigma(T_i) \leq \alpha \sigma^* $ for all $1\leq i \leq k$. Here, $d$ denotes the symmetric difference of edge sets of two triangulations, and $\sigma^*$ denotes the best quality of triangulations of $P$, e.g., the minimum Euclidean length.

\hspace{4mm} As our main result, we provide a $\mathrm{poly}(n,k)$-time approximation algorithm for the DNT problem that returns a collection of $k$ distinct triangulations whose diversity is at least $1 - \Theta(1/k)$ of the optimal, and each triangulation satisfies the quality constraint. This is accomplished by studying \emph{bi-criteria triangulations} (BCT), which are triangulations that simultaneously optimize two criteria, a topic of independent interest. We complement our approximation algorithms by showing that the DNT problem and the BCT problem are NP-hard.

\hspace{4mm}Finally, for the version where diversity is defined as $\min_{i < j} d(T_i,T_j)
$, we show a reduction from the problem of computing optimal Hamming codes, and provide an $n^{O(k)}$-time $\tfrac12$-approximation algorithm.
This improves over the naive ${C_{n-2} \choose k} \approx 2^{O(nk)}$ time bound for enumerating all $k$-tuples among the triangulations of a simple $n$-gon, where $C_n$ denotes the $n$-th Catalan number.
\end{abstract}

\tableofcontents

\clearpage
\input{01_intro}
\input{02_Results}

\input{03_Technical_Overview_sum_dnt}
\input{04_Conclusion}

\bibliographystyle{splncs04}
\bibliography{reduced-ref.bib}

\clearpage
\appendix
\input{05_DNT_Hardness}
\clearpage
\input{06_BCT_Hardness}

\clearpage
\input{07_BCT_details}
\clearpage
\input{08_proof_dnt_special}
\clearpage
\input{09_aqnd}
\clearpage
\input{10_Max_Min}

\end{document}

%% file: 01_intro.tex
\section{Introduction}
\label{sec:introduction}

There has been considerable recent interest in computing diverse solutions to optimization problems, also called the diverseX paradigm. Instead of computing a single (approximately) optimal solution to a (perhaps hard) optimization problem, this paradigm asks to compute a set of $k \geq 2$ many good solutions that are as different from each other as possible. There are several motivations for this, such as providing more choices to the user~\cite{baste2019fpt}, fairness~\cite{aumuller2020fair,aumuller2021fair}, robustness and security~\cite{gao2022obtaining}, and portfolio optimization~\cite{drygala2024data}. 

Computing diverse solutions to optimization problems have been studied for problems such as vertex cover and hitting set ~\cite{baste2022diversity,baste2019fpt}, matchings ~\cite{fomin2020diverse}, shortest paths ~\cite{hanaka2022computing,funayama2024parameterized}, spanning trees ~\cite{hanaka2021finding}, approximate minimum spanning trees, global min cuts, matchings, short paths ~\cite{gao2022obtaining}, $s$-$t$ min cuts~\cite{de2023finding}, and SAT~\cite{misra2024parameterized,austrin2024geometry}.

In this work, we focus on computing diverse solutions in computational geometry (CG). One of the first topics in a CG course is the problem of computing triangulations of a given polygon or a point set. What if we want to compute $k$ very different triangulations? When $k=2$, we want to compute the diametral pair of triangulations according to some metric on the space of triangulations. This problem has been studied qualitatively when the metric is the flip distance between triangulations: even the simple case of convex polygons turns out to be fascinating, as computing the triangulations of a convex polygon with maximal flip distance is related to the rotation distance between binary search trees and the diameter of associahedra \cite{pournin2014diameter}. However, for arbitrary polygons even computing the flip distance between two given triangulations is NP-complete~\cite{aichholzer2015flip}. In this work we focus on another natural metric, the symmetric difference of the edge sets of two triangulations. We will consider non-convex polygons as well as obtaining $k>2$ solutions. Moreover, since many applications require triangulations that are ``nice'' with respect to some measure, we will study the \emph{diverse and nice triangulations} problem where we want to return a maximally diverse set of $k$ nice triangulations.

The diverse and nice triangulations problem is not only of theoretical interest but also of potential practical significance. 
Decomposing a complex shape into simple pieces is required for many applications ranging from computer graphics and vision~\cite{deberg2008} to numerical analysis~\cite{lawson1977software}. 
Presenting the user with many \emph{different} looking triangulations not only increases available choices but may also increase the robustness of these procedures.
As another example, consider an application in numerical analysis - the finite element method used to solve differential equations~\cite{babuvska1976angle}. 
A result achieved after averaging the computation over diverse triangulations may provide a more robust, more domain-dependent and less triangulation-dependent answer, as opposed to performing the computation over just one triangulation, or a few similar triangulations. 

Moreover, in many application certain triangulations are preferred over others. Well-studied examples include the \emph{minimum (Euclidean) length triangulation, the Delaunay triangulation that maximizes the minimum angle}, and others. In this setting, we want triangulations that are not only diverse, but also ``nice'' with respect to such quality measures.
This leads us to formulate the \emph{diverse and nice triangulations} problem, which we formally define next.

\subsection{Problem Statement}\label{sec: problem-statement}
Unless otherwise stated, we assume that $P$ is a simple $n$-gon, i.e., a polygon without holes. Define a triangulation of $P$ as a maximal set of non-crossing diagonals of $P$. We first define our diversity measure for the triangulations.

\begin{definition}[Diversity Measure]\label{def: diversity-measure} Given triangulations $ T_1, \ldots, T_k $ of $P$, we consider two diversity measures:
    \[ \SumSD[T_{1},\dots,T_{k}] \overset{\text{def}}{=} \sum\limits_{1 \leq i < j \leq k} |T_{i} \Delta T_{j}|, \ \ \ \text{and} \ \ \ \MinSD[T_{1},\dots,T_{k}] \overset{\text{def}}{=} \min\limits_{1 \leq i < j \leq k} |T_{i} \Delta T_{j}| \]
\end{definition}

Our goal is to develop algorithms that maximize these two diversity measures. Note that maximizing $\SumSD$ is equivalent to maximizing the \emph{average} pairwise distance between the triangulations. In addition, we will not only care about the diversity among a collection of triangulations, but with different applications in mind, also care about the \emph{quality of each triangulation}. Let $\Tcal$ denote the set of all triangulations of $P$. Given a quality measure $\sigma: \Tcal \rightarrow \mathbb{R}_{\geq_0}$ and an approximation factor $\alpha \geq 1$, we say a triangulation $T$ is \emph{nice} if $ \sigma(T) \leq \alpha \sigma^* $, where by convention, we assume that $\sigma^*=\min_{T \in \Tcal} \sigma(T)$ is the quality of a best triangulation (w.r.t. $\sigma$) of $P$.

\begin{definition}[Sum-DNT and Min-DNT] \label{def: dnt}
 Input: A simple $n$-gon $P$ without holes, $k \geq 2 $, a quality measure $\sigma$ and an $\alpha \geq 1$. Output: A collection of $k$ distinct triangulations $T_1,\dots,T_k \in \T$, if they exist, such that the following holds:
\begin{enumerate}
    \item ($\alpha$-optimal) For every $i\in \{1,\ldots,k\}$,  $\sigma(T_{i}) \leq \alpha \cdot \sigma^*$.
    These $\alpha$-optimal triangulations are called ``nice'' triangulations.
    \item (Maximally diverse) For any set of $k$ distinct $\alpha$-optimal triangulations $ T^{'}_1 $, ..., $T^{'}_k\in\T$ 
    (i.e, satisfying 1 above),
    $ \SumSD[T_1,\ldots,T_k] \geq \SumSD[T'_1,\ldots,T'_k] $ (Sum-DNT) or  $ \MinSD[T_1,\ldots,T_k] \geq \MinSD[T'_1,\ldots,T'_k] $ (Min-DNT).
\end{enumerate}
\end{definition}

Note that the DNT problem comes with two measures - the diversity measure $\SumSD$ or $\MinSD$ on a collection of $k$ triangulations, and the quality measure $\sigma$ for each triangulation.
If the quality measure $\sigma$ is not considered (we just set it to the constant function), omit the ``N'' and call these variants of problem the Diverse Triangulations (Sum-DT or Min-DT) problems.

\subsection{Related Work}\label{sec: related-work}

As far as we know, we are the first to study the problem of finding diverse triangulations. In the field of computational geometry, the only existing work on finding diverse solutions is the recent interesting work by Klute and van Kreveld~\cite{klute2022fully} that considers diverse sets of geometric objects, such as polygons and point sets, with several diversity metrics. The authors investigate the maximum size of a fully diverse set, defined as a set of $k$ objects such that $\min_{i \neq j} d(x_i,x_j)$ is at least a constant fraction of the diameter of the space. The focus is on quantitative upper and lower bounds on the sizes of such sets, and not algorithmic. The authors mention a simple randomized algorithm that samples an object and adds it to the collection if it is sufficiently far away from all current objects in the collection, but the runtime is not analyzed. Apart from bounding the number of samples needed, an additional issue in implementing this algorithm for the nice triangulation problem is that while sampling a triangulation at random can be done in polynomial time~\cite{epstein1994generating}, it is not clear how to sample uniformly a nice triangulation (e.g., one with Euclidean weight at most twice that of the minimum) in polynomial time. For a survey on quality measures of triangulation, see~\cite{de2010triangulations}.

A related but fundamentally different problem is that of \emph{diversity optimization}, or \emph{dispersion}, that has been studied extensively~\cite{erkut1990discrete,BorodinJLY17,abboud2022improved}. 
In the latter problem, one wants to obtain a set of $k$ maximally dispersed points from a metric space of $N$ points.
Note that for our problem, the metric space is the space of all triangulations, and so $\textrm{poly}(N,k)$ algorithms for dispersion translate directly to a $\textrm{exp}(n)\textrm{poly}(k)$ algorithms for the diverse solutions.

The concept of diversity has also been explored \emph{within a solution}. Recent examples include work on packing~\cite{galvez2022approximation}, diverse convex and Voronoi partitions~\cite{van2021diverse}, and the dispersive art gallery problem~\cite{rieck2024dispersive}.

%% file: 02_Results.tex
\section{Our Results}
\label{sec: results}

We state our results for the Sum-DNT problem first, followed by the results for the Min-DNT problem.

\subsection{Sum-DNT Results}

\noindent\textbf{Hardness:} Observe that a successful output of the DNT problem consist of $k$ triangulations and hence takes space $\Omega(kn)$, which is not polynomial in the size of input already.
However, we show that even the decision version of the DNT problem (DDNT) is NP-hard; see \Cref{sec: DDNT-Hardness}. For this hardness, our quality measure $\sigma(T)$ on a triangulation $T$ is the sum of the Euclidean lengths of all the $(n-3)$ diagonals in the triangulation. We call this quantity $\sigma(T)$ the Euclidean length of the triangulation.
\begin{restatable}[NP-hardness of DDNT]{theorem}{ResDNTHardness}\label{thm: DDNT-Hardness}
There exists a simple polygon $P$ with $n$ vertices,  an integer $D \geq 1$, and an $\alpha' \in (1, \infty)$ such that no $\text{poly}(n, \log k)$-time algorithm can decide whether $P$ has $k$ distinct triangulations $T_1, \ldots, T_k$ such that $\SumSD[T_1, \ldots, T_k] \geq D$ and $\sigma(T_i) \leq \alpha' \sigma^*$ for all $i \in [k]$, unless $\classP=\classNP$. Here, $\sigma^*$ denotes the minimum Euclidean length of a triangulation of $P$.
\end{restatable}

\vspace{2mm}\noindent\textbf{Decomposable Quality Measure:} In light of the hardness above, our first contribution to the DNT problem is identifying a class of quality measures that we call \emph{decomposable}. It will turn out that not only does this class contain most of the quality measures on triangulations studied in the computational geometry community, but we can also develop approximation algorithms for these quality measures!
Assume first that a quality measure $\sigma'$ for a single edge, or a single triangle is given; e.g., $\sigma'$ measures the length of an edge, or the minimum angle of a triangle.
A \emph{decomposable} quality measure $\sigma$ on a triangulation $T$ is one that is either a min, max or sum, of the measure $\sigma'$ over the edges or triangles in $T$.

\begin{definition}[Decomposable Quality Measures]\label{def:decomposable}
 Let $\odot \in \{\sum, \min, \max\}$, and let $\sigma'$ be a nonnegative function defined on the set of diagonals or triangles of the polygon.
 A quality measure $\sigma: \T \rightarrow \mathbb{R}_{\geq 0}$ is said to be edge-decomposable if $\sigma(T)=\odot_{e \in T} \sigma'(e)$ for any triangulation $T$, and triangle-decomposable if $\sigma(T)=\odot_{t \in tr(T)} \sigma'(t)$ for any triangulation $T \in \T$.
\end{definition}
By a decomposable measure we will mean a measure that is either edge-decomposable or triangle-decomposable.
Examples of decomposable quality measures that have been considered extensively in the literature are the Euclidean length \cite{mulzer2008minimum,remy2009quasi}, maximum length (of the $(n-3)$ diagonals) \cite{edelsbrunner1993quadratic}, minimum length (of the $(n-3)$ diagonals) \cite{fekete2014computing}, maximum angle of the triangulation \cite{edelsbrunner1990n}, and the minimum angle of the triangulation, a measure that is maximized by the popular Delaunay triangulation. There is also research on triangulations that are close-to-Delaunay~\cite{van2021near}; see \Cref{fig:delaunaymeasures-app} in \Cref{sec: bct-hardness}.

\vspace{2mm}\noindent\textbf{Approximation Algorithms:} We now present our main algorithmic results for the DNT problem with respect to the $\SumSD$ measure.
All of our algorithms will have an approximation factor of \defi{$\beta := \max\left\{ \frac{1}{2}, 1 - \frac{2}{k+1} \right\}$} for the diversity. $\beta$ equals $1/2$ when $k \leq 3$, and approaches $1$ as $k \rightarrow \infty$.

\begin{restatable}[Algorithms for Sum-DNT]{theorem}{dnttheorem}
    \label{thm: dnttheorem}
    Let $\sigma$ be any decomposable measure. Let $P$, $k$ and $\alpha$ be given as in Definition 2 of the DNT problem, and let $T_1,\ldots,T_k $ and $ T'_1,\ldots,T'_k $ in the following denote distinct triangulations. Then,
        \begin{enumerate}
            \item{\label{dnt: general}} For $ \alpha \in (1, \infty) $, there exist an $O(n^5k^5\log k)$-time algorithm that returns $k$ $\alpha$-optimal triangulations $T_{1},\dots,T_{k}$ such that $ \SumSD[T_{1},\dots,T_{k}] $ is at least $ \beta \cdot \SumSD[T'_{1},\dots,T'_{k}]$ for any $\alpha$-optimal triangulations $T'_{1},\dots,T'_{k}$.

            \item{\label{dnt: fptas}} For $ \alpha \in (1, \infty) $, there exists an algorithm that runs in time $ O \left( \varepsilon^{-2}\cdot n^5 k^3\log k \right) $ and returns $k$ $ \alpha(1+\varepsilon) $-optimal triangulations $T_{1},\dots,T_{k}$ such that $\SumSD[T_{1},\dots,T_{k}] $ is at least $ \beta \cdot \SumSD[T'_{1},\dots,T'_{k}]$ for any $\alpha$-optimal triangulations $T'_{1},\dots,T'_{k}$.
        \end{enumerate}
\end{restatable}

For some special cases, the Sum-DNT problem has faster algorithms or even admits PTAS. We provide proofs of these special cases in \Cref{sec: dnt-special}.
\begin{restatable}[Special Cases for Sum-DNT]{theorem}{ResDNTSpecial}\label{thm: dnt-special}
    \begin{enumerate}
        \item{\label{dnt: special}} If $\sigma$ is decomposable and $\alpha = 1 $, there exists an $ O(n^3k^3\log k) $-time algorithm for Sum-DNT that returns $k$ triangulations $T_{1},\dots,T_{k}$ with diversity at least $\beta$ of optimal.
        \item When $P$ is a convex polygon, there is an algorithm for Sum-DT that runs in time $2^{O(1/\varepsilon^2)}n^5k^5$ and returns $k$ distinct triangulations whose diversity is at least $ (1-\varepsilon) $ of the optimal.
        \item If $\sigma$ is Delaunay measure and $\alpha=1$, then there is an algorithm that runs in time $k^3n^{O(1/\varepsilon)}$ and returns $k$ distinct Delaunay triangulations of $P$ whose diversity is at least $(1-\varepsilon)$ of the optimal.
    \end{enumerate}
\end{restatable}

\subsection{Min-DNT Results}

We now state our results on triangulations that maximize the minimum distance between two triangulations, denoted by $\MinSD$. In general, the max-min version of dispersion is considered harder than the max-sum version, and some hardness results for matroids are given by Fomin et al.~\cite{fomin2023diverse}. For this version, we will ignore the niceness constraint, and only focus on the min-DT problem.

For simple polygons, we relate the min-DT problem to the problem of computing $A_q(n,d)$, which is the maximum number of $q$-ary codewords of length $n$ with pairwise Hamming distance at least $d$. If we let $D_q(n,m,d)$ denote the decision version $A_q(n,d) \geq m$, then we can show the following interesting result, whose proof can be found in \Cref{sec: aqnd-proof}.

\begin{restatable}[Reduction to Hamming Codes]{theorem}{cctrihardness}\label{thm:codes}
   Assume there is an algorithm that, given a polygon with $n$ vertices and an integer $k = O(n)$, runs in time in $\mathrm{poly}(n)$ and outputs $k$ diverse triangulations maximizing $\min_{\normalfont\text{SD}}$.
   Then there is an algorithm for computing $A_2(n, d)$ for any $d > n/2$ in time $\mathrm{poly}(n)$.
\end{restatable} 

As far as we know, computing $A_2(n,d)$ is still open, and only a limited number of instances are currently known, e.g., see~\cite{ostergard2011size}. 
Now, we provide an algorithm for the min-DT problem, whose proof can be found in \Cref{sec: proof-min-dt}.
\begin{restatable}[Algorithm for Min-DT]{theorem}{mindttheorem}
    \label{thm: mindttheorem}
    Let $P$ and $k$ be given as in Definition 2 of the DNT problem.
    Define $r=2(n-3)-\frac{d_{\mathrm{OPT}}}{2}$, where $d_{\mathrm{OPT}}$ denotes optimal diversity in $\MinSD$ measure.
    Then, there exists an $ r^{O(k)} $-time algorithm that outputs $k$ triangulations such that
    \[
        \MinSD[T_1,\ldots,T_k] \geq \frac{1}{2}\cdot \MinSD[T'_1,\ldots,T'_k]
    \]
    for any triangulations $T'_1,\ldots,T'_k$ of $P$.
\end{restatable}
Note that the algorithm in \Cref{thm: mindttheorem} runs fast if $ d_{\mathrm{OPT}} $ is large, i.e., if $P$ has very diverse triangulations, our algorithm can find those faster.

%% file: 03_Technical_Overview_sum_dnt.tex
\section{Technical Overview for Sum-DNT: Enter Bicriteria Triangulations}\label{sec: technical-overview-sum-dnt}

In the remainder of the main body of this abstract, we will focus on the Sum-DNT version, since we consider the polynomial time approximation algorithms for this version as our main result. Proofs of the hardness of Sum-DNT (\Cref{thm: DDNT-Hardness}) can be found in~\Cref{sec: DDNT-Hardness}, the relation of min-DT to Hamming codes (\Cref{thm:codes}) can be found in~\Cref{sec: aqnd-proof} and the algorithmic results on min-DNT version (\Cref{thm: mindttheorem}) can be found in~\Cref{sec: proof-min-dt}.

In the Sum-DNT problem (\Cref{def: dnt}), the diversity is the sum, or the average, of the pairwise distances between the $k$ triangulations. Let us recall the classical problem of dispersion, where one wants to obtain a set of $k$ maximally dispersed (w.r.t. either the sum or the minimum of the pairwise distances) points $x_1,\cdots,x_k$ from a metric space $\mathcal{M}$ of $N$ points. While $\classNP$-complete, an application of the farthest insertion method gives an approximation algorithm for both the sum and the minimum versions of this problem. Set $x_1$ to be an arbitrary point in $\mathcal{M}$, and iteratively select $x_i =\argmax_{p \in \mathcal{M}} \sum_{1\leq j \leq i-1} d(p,x_j)$ for the sum version, and $x_i =\argmax_{p \in \mathcal{M}} \min_{1\leq j \leq i-1} d(p,x_j)$ for the minimum version. Both algorithms give a $(1/2)$-approximation guarantee to their respective objective functions \cite{BorodinJLY17,ravi1994heuristic}, and clearly run in time $\mathrm{poly}(N,k)$.

For now let us aim for a $(1/2)$-approximation factor for the Sum-DNT problem; later we will show how to boost it to $\betak$. Let $\mathcal{T}$ be the space of all triangulations of a polygon with $n$ vertices, and $\mathcal{T}_\alpha$ the space of $\alpha$-optimal, or nice triangulations as in Definition 2. Applying the dispersion algorithms (for the sum of pairwise distances) from \cite{BorodinJLY17} translates to the following: having found $i$ triangulations $T_1,\cdots,T_i$, we set $T_{i+1}$ as the triangulation $T \in \mathcal{T}_\alpha$ that maximizes $\sum_{j=1}^{i} |T_j\Delta T|$. This would give a $(1/2)$-approximation to the $\SumSD$.

\subsection{Bi-Criteria Triangulations}\label{sec: bct-introduction}

It turns out that finding $T_{i+1}$ is a special case of the following problem.
\begin{definition}[$\BCT(\cdot,\cdot)$] \label{def: bct}
    Given a polygon $P$, two measures \emph{weight} and \emph{quality} $ w,\sigma: \cT \rightarrow \mathbbm{R}_{\geq0} $ and a bound $ B \geq 0 $, the Bi-Criteria Triangulation is a solution to the following program:
    \begin{equation*}
    \argmin\{w(T): T \in \mathcal{T} \text{ and } \sigma(T) \leq B\}.
    \end{equation*}
    When $B=\alpha\sigma^*$, the obtained solution is called $\alpha$-optimal $\BCT$.
    Additionally, $k$ distinct triangulations $T_1,\ldots,T_k$ of $P$ are said to be the $k$-best enumeration for BCT provided that 1) $ w(T_i) \leq B  $ for each $i\in[k]$ and 2) $ w(T_1) \leq \ldots \leq w(T_k) \leq w(T') $ for any $ T' \in \mathcal{T}\setminus\{T_1,\ldots,T_k\} $ such that $ \sigma(T') \leq B $.
    The notion of $k$-best $\alpha$-optimal $\BCT$s are defined analogously.
\end{definition}

\noindent\textbf{Remark.} Although the BCT is defined as a minimization problem with ``$\leq$'', $\min$ and ``$\leq$'', respectively, can be replaced with $\max$ and $\geq$ independently. 

When $k=1$, the BCT problem asks for a \emph{triangulation that is simultaneously good w.r.t. two measures}. Surprisingly, this problem has not appeared in literature, even though historically very related questions have been asked, such as whether the Delaunay triangulation has low Euclidean length~\cite{lloyd1977triangulations}.

\vspace{2mm}\noindent\textbf{Reduction of Farthest Insertion to BCT.} Assuming one has an algorithm for BCT, we show how to compute $T_{i+1}$ 
(which is $\argmax_{T \in \mathcal{T}_\alpha} \sum_{j=1}^{i} |T_j\Delta T| $).
This follows from a general framework by \cite{hanaka2023framework,gao2022obtaining}, which we translate to our setting in the following proposition.

\begin{restatable}{proposition}{RestMinCE}\label{prop: min-ce}
Let $ T_1, \dots, T_i $ be $\alpha$-optimal triangulations of $P$. For any $T\in\cT_{\alpha}$, define $ w_i(T) := \sum_{e\in T}\sum_{j=1}^i \mathbbm{1}(e\in T_j)$.
Then, \begin{equation}\label{eq: farthest-bct-relationship}
    \argmax_{T \in \cT_{\alpha}} \sum_{j=1}^{i} |T_j\Delta T| = \argmin_{T\in\cT_{\alpha}}{w_i(T)}.
\end{equation}
\end{restatable}
\begin{proof}
Let $ T_1,\ldots,T_i $ be $\alpha$-optimal triangulations. Then,
\[
\begin{split}
    \max_{T \in \cT_{\alpha}} \sum_{j=1}^i \lvert T \Delta T_j \rvert
        &= \max_{T \in \cT_{\alpha}}\sum_{j=1}^i \Big( 2(n - 3) - 2|T \cap T_j| \Big) \\
        &= \max_{T \in \cT_{\alpha}} \Big( 2(n-3)i - 2\sum_{j=1}^i |T \cap T_j| \Big) \\
        &= 2(n - 3)i - 2\min_{T \in \cT_{\alpha}}\sum_{j=1}^i |T \cap T_j| \\
        &= 2(n-3)i - 2\min_{T \in \cT_{\alpha}}\sum_{e\in T}\sum_{j=1}^i \mathbbm{1}(e\in T_j),
\end{split}
\]
since every triangulation of a simple $n$-gon comprises $(n-3)$ edges. 
\end{proof}
\noindent Note that the right hand side of \Cref{eq: farthest-bct-relationship} is exactly the $\alpha$-optimal BCT with measures $w_i$ and $\sigma$. Hence, computing a farthest $\alpha$-optimal triangulation is equivalent to solving an $\alpha$-optimal $\BCT$.

We now illustrate the algorithm. For the time being, we denote $ f(w_i,\sigma) $ and $f_k(w_i,\sigma)$, respectively, the running times for solving the $\alpha$-optimal $\BCT(w_i,\sigma)$ and the $k$-best $\alpha$-optimal $\BCT(w_i,\sigma)$.

Initially, set the weight $w_0$ of each allowed diagonal of $P$ to 0, $\sigma$ as in the input to the Sum-DNT problem, and $B$ as $\alpha \sigma^*$. Then, solve the associated BCT problem, i.e., $ \BCT(w_0,\sigma) $, obtaining $T_1$.
Next, for every $ e \in T_1$, increase the weight of $e$ by 1, solve $ \BCT(w_1,\sigma) $, and call the output triangulation $T_2$.
Increase the weight of each $ e \in T_2$ by 1, solve  $ \BCT(w_2,\sigma) $, and so on.
Note that during the process a copy of one of the preceding triangulations might be obtained again, e.g., when all allowed edges of $P$ have appeared exactly the same number of times.
To avoid duplicates, we utilize the $k$-best enumeration procedure: given $i$ triangulations, instead of finding a single BCT, find $(i+1)$ distinct BCTs, $T'_1, \ldots, T'_{i+1} $ such that $ w_{i}(T'_1) \leq \ldots \leq w_{i}(T'_{i+1}) \leq w_{i}(T') $ for all $ T' \in \mathcal{T}_{\alpha}\setminus\{T'_1,\ldots,T'_i\} $. Then, one of them must be distinct from all previous triangulations, and we take this triangulation as our $T_{i+1}$.
Therefore, to obtain $\tfrac12$-approximation in diversity, it is enough to run the $k$-best $\alpha$-optimal $\BCT$ program $k$ times, which incurs a total running time of $ O(kf_k(w_k,\sigma))$.

\vspace{2mm}\noindent\textbf{Improving $1/2$ to $\betak$.} We now illustrate how to obtain the improved $ \betak $ factor for large $k$.
It turns out that the symmetric difference of two sets is a \emph{negative type} metric~\cite{duin2009dissimilarity}, and for such metrics it is possible to obtain the improved approximation factor by using a local-search based swapping algorithm ~\cite{cevallos2019improved,hanaka2023framework}.
This swapping algorithm begins with any set of $k$ solutions, say $ S_0 $.
At step $i$, then, the algorithm finds two triangulations $ T^* $ from outside $ S_i $ and $T_j$ inside $S_i$ that maximize the diversity of the current solution set when they are swapped.
Note that such $T^*$ can be found by comparing the farthest triangulations of $ S_i \setminus \{T_j\} $ for all $j\in[k]$.
Therefore, each step of the swapping algorithm can be done by performing farthest insertions $k$ times.
Cevallos et al.~\cite{cevallos2019improved} guarantee that one needs $O(k\log k)$ iterations of each step to obtain the desired diversity.
Consequently, the overall running time for finding $k$ triangulations with diversity of $\betak$ is $O(f_k(w_k,\sigma) k^2 \log k)$.

\subsection{Results on BCT}

With the general reduction above, we now state our algorithms for the BCT problem that will deliver the promised approximation algorithms for the Sum-DNT problem in \Cref{thm: dnttheorem} by showing $f_k(w_k,\sigma) \in \mathrm{poly}(n,k)$.
Before we state the approximation algorithms for the BCT problem, and since the BCT problem is interesting on its own, we mention a potentially important but tangential (for our purposes) observation, where proof can be found in \Cref{sec: bct-hardness}.

\begin{restatable}[Hardness of BCT]{theorem}{ResBCTHardness}\label{thm:bcthardness}
    Let the weight $w$ on a diagonal equal its Euclidean length. Let $\sigma$ be any of the near-Delaunay measures in \cite{van2021near}, and $B \geq 0$ be given. Then solving the $\BCT$ problem w.r.t. $w$ and $\sigma$ is $\nphard$.
\end{restatable}

Despite the hardness result above, we now state the promised positive algorithms for BCT. Fortunately, it turns out that the BCT problems admits pseudo-polynomial time algorithms if both weight and quality measures $w$ and $\sigma$ are decomposable.
Furthermore, we also prove that the BCT problem admits an FPTAS. These two results will deliver the two results stated in \Cref{thm: dnttheorem}.

\begin{restatable}[Algorithms for BCT]{theorem}{ResBCTAlgorithms}\label{thm: bct-algorithms}
    Consider the problem $\BCT(w,\sigma)$ problem with bound $B \geq 0$, where both $w$ and $\sigma$ are decomposable. Then for any integer $k\geq1$ the following hold:
    \begin{enumerate}
        \item\label{itm: bct-algorithm-integer-measure}
    \textbf{(Integer‐valued measure)} 
    The $k$-best enumeration for $\BCT(w,\sigma)$ can be solved in time
    $ O\bigl((M+1)^2\,k\,n^3\bigr)$,
    where $M=W$ in the case $w(\cdot)$ is integer-valued and $w(T)\in[0,W]$ for some integer $W\geq0$ and $M=B$ in the case $\sigma(\cdot)$ is integer-valued.
    \item\label{itm: bct-algorithm-fptas} \textbf{(FPTAS)} For any $\varepsilon>0$, let $ T^* = \BCT(w,\sigma)$. Then, there is an $O(\varepsilon^{-2} n^5 )$-time algorithm that returns a triangulation $\tilde{T}$ such that $w(\tilde{T}) \leq w(T^*) $ and $ \sigma(\tilde{T})\leq (1+\varepsilon)B $. In particular, its $k$-best version runs in time $O\left( \varepsilon^{-2} n^5 k \right)$.
    \end{enumerate}
\end{restatable}

\begin{proof}
\begin{enumerate}
    \item\textbf{(Integer-valued measure)} First, we prove for the cases where both $w(\cdot)$ and $\sigma(\cdot)$ are additive, triangle-decomposable measures. The remaining cases follow by an analogous argument, which we outline at the end of the proof.

    Assume that $w(\cdot)$ and $\sigma(\cdot)$ are additive, triangle-decomposable measures. Let $W$ be a fixed nonnegative integer.
    
    \hspace{4mm}\textbf{Case 1: $w(\cdot)$ is integer-valued and $w(T)\in[0,W]$.} Assume that $w(\cdot)$ is integer‑valued and satisfies $w(T)\in[0,W]$.
    Let $P[i:j]$ denote the closed chain of vertices of $P$, from $i$ to $j$, then returning to $i$ in a counterclockwise direction, and let $\cT[i:j]$ denote the collection of all triangulations of $P[i:j]$.
    Let $OPT_k(W',i,j)$ denote the sorted multiset of the $k$ smallest $\sigma(T)$ among $T\in\cT[i:j]$ with $w(T)=W'$.
    If fewer than $k$ such triangulations exist, we pad with $\infty$.
    Whenever $\overline{ij}$ is not an edge or diagonal of $P$, we set $OPT_k(W',i,j)=\{\infty,\ldots,\infty\}$; if $j=i+2$ and $\overline{ij}$ \emph{is} a diagonal of $P$, we set $OPT_k(W',i,i+2)=\{\sigma(\triangle(i,i+1,i+2)),\infty,\ldots,\infty \}$.

    \hspace{4mm}For $j>i+2$, since every triangulation in $\cT[i:j]$ chooses some vertex $m \in [i+2,j-2]$ to form the triangle $\triangle imj$, splitting at $m$ yields the following:
    \begin{equation}\label{eq: k-best}
        OPT_k(W',i,j) = \theta_k \Bigg(\bigcup_{m,W'_1,W'_2} \Big(OPT_k(W'_1,i,m) + OPT_k(W'_2, m,j) + \sigma(\triangle imj)\Big)\Bigg),
    \end{equation}
    where $\theta_k (A) $ denotes the smallest $k$ elements of $A$, the union is taken over all $ m \in [i+2, j-2] $ and all $ W'_1,W'_2 \in [0,W'-w(\triangle imj)] $ such that $ W'_1 + W'_2 = W'-w(\triangle imj) $, and $ A+B+c $ denotes $ \{ a+b+c\mid a\in A, b\in B \} $. Thus, dynamic programming on triples $(W',i,j)$ with $W'\in[0,W]$ and $1\leq i < j \leq n$ fills a table of size $(W+1)n^2$.

    \hspace{4mm} Fix $W'$, $i$ and $j$.
    For any $m$, $W'_1$ and $W'_2$, we can select the $k$ smallest elements of $ OPT_k(W'_1,i,m) + OPT_k(W'_2, m,j) + \sigma(\triangle imj) $ in time $O(k)$ by using the selection algorithm in \cite{frederickson1982complexity}.
    Thus, to evaluate the right-hand side of \Cref{eq: k-best}, we loop over all choices of $W'_1$, $W'_2$, and $m$, compute the $k$ smallest sums for each triple, and then select the smallest $k$ elements from their union. This entire process runs in $O(W'kn)$ time, i.e., one cell of the table can be filled in time $O(W'kn)$.
    Since the table consists of $ O(Wn^2) $ cells, the overall time bound to fill the entire table is $O(W^2kn^3)$.
    Finally, starting from $W'=0$, scan the cells $ OPT_k(W',1,n) $ for values $\leq B$, and return the desired set of up to $k$ smallest weights.
    If fewer than $k$ such values exist, we report that the given polygon has less than $k$ nice triangulations.

    \hspace{4mm}As usual, we can compute the actual triangulations with standard bookkeeping, without affecting the overall running time. We omit the details here.

    \vspace{2mm}\noindent\textbf{Case 2: $\sigma(\cdot)$ is integer-valued.} Suppose $\sigma(\cdot)$ takes integer values. The $k$-best enumeration can be carried out similarly as in Case 1, with only an additional $O(k)$ factor in the running time. Therefore, we concentrate here on finding a single optimal solution.

    \hspace{4mm}Let $OPT(B',i,j)$ denote the minimum weight of triangulation of $P[i:j]$ whose quality is at most $B'$:
\[ OPT(B',i,j) = \min\{ w(T) \mid \sigma(T) \leq B', T\in\cT[i:j] \}. \]
By convention, we define $ OPT(B',i,j) = w(\triangle(i,i+1,i+2)) $ when $j=i+2$.
When $j>i+2$, then $ OPT(B',i,j) $ can be found by using the following recurrence relation:
\begin{equation*}
    OPT(B',i,j) = \min\limits_{\substack{m \in [i+2, j-2] \\ B'_1,B'_2 \in [0, B'-\sigma(\triangle imj)], \\ B'_1 + B'_2 = B'-\sigma(\triangle imj)}} OPT(B'_1,i,m) + w(\triangle imj) + OPT(B'_2, m,j)
\end{equation*}

When any of $\overline{im}$, $\overline{mj}$ and $\overline{mj}$ is not an allowed diagonal of $P$, or $B'<\sigma(\triangle imj)$, we set $OPT(B',i,j) = \infty$.
Then, it is easy to check that $ OPT(B,1,n) $ can be computed in time $ O((B+1)^2 n^3) $, which gives an overall time bound of $ O((B+1)^2 k n^3) $.

\vspace{2mm}\noindent\textbf{Other Cases.} We end the proof by noting similar dynamic programs can be easily constructed when $w(\cdot)$ and $\sigma(\cdot)$ are not additive, but are decomposable with $\min$ or $\max$.
For example, let $w(\cdot)$ and $\sigma(\cdot)$ are both edge-decomposable with respect to $\min$, and $w(\cdot)$ ranges over $[0,W]$, let $ OPT(W',i,j) $ denote the best quality of a triangulation of $P[i:j]$, and define $ OPT(W',i,j)=\infty $ when $ j=i+1 $.
When $j>i+1$, then $OPT(W',i,j)$ can be computed by using the following recurrence relation:
\begin{equation*}
    OPT(W',i,j) = \min_{m, W'_1,W'_2} \Big\{ \min \{OPT(W'_1,i,m),\sigma(\triangle imj), OPT(W'_2, m,j)\} \Big\},
\end{equation*}
where the first minimum in the equation above is taken over all $m \in [i+1, j-1] $ and $ W'_1,W'_2 \in [0, W'-(\triangle imj)] $ such that $ W' = \min \{W'_1,W'_2,w(\triangle imj)\} $.
The other cases can be similarly handled, we omit the details.

    \item \textbf{(FPTAS)} Assume that $ w(\cdot) $, $ \sigma(\cdot) $, and a budget $B\geq0$ are given. For brevity, we assume that both $ w(\cdot) $ and $ \sigma(\cdot) $ are additive and triangle-decomposable. The other cases can be handled similarly. Also, we focus on finding single best solution since the $k$-best enumeration can be handled analogously with only extra $O(k)$ factor in the running time as in Case 1 above.
    
    \hspace{4mm} Assume that $k=1$. The main idea of the algorithm is to scale down the quality measures of the allowed triangles in the polygon, as well as the given budget, to integers of size at most $\mathrm{poly}(n, \varepsilon)$.  
    We then run the BCT algorithm described above in Case 2, over these scaled weights and budget. Since the running time of the algorithm in Case 2 is $O((B+1)^2 n^3)$, the BCT problem can be solved in cubic time when $B=0$. Thus, we focus ourselves on the case where $ B > 0$. 

    \hspace{4mm}Assume that $ B > 0 $.
    Given $ \varepsilon > 0 $, define
        \begin{equation}
            \tilde{\sigma}(t) := \left\lfloor \frac{n-2}{\varepsilon B} \cdot \sigma(t) \right\rfloor, \qquad \tilde{B} := \left\lfloor \frac{n-2}{\varepsilon} \right\rfloor,
        \end{equation}
     where $t$ is any allowed triangle of $P$.
     Then we solve $ \BCT(w, \tilde{\sigma}) $ with bound $\tilde{B}$, i.e.,
     \begin{equation*}\label{eq: fptas-constraint}
         \argmin\{w(T):T \in \mathcal{T}\text{ and }\tilde{\sigma}(T) \,\,\leq\,\, \tilde{B}\}
     \end{equation*}
    which can be done in time $O(\tilde{B}^2 \cdot n^3) = O(\varepsilon^{-2} \cdot n^5)$ by using the BCT algorithm in Case 2.

    \hspace{4mm} Let $ \tilde{T} $ be an output triangulation, and $\tstar$ be an optimal triangulation.
    We now prove that the $\tilde{T}$ and $\tstar$ satisfy the desired conditions, i.e.,  $ \sigma(\tilde{T}) \leq (1+\varepsilon)B $ and $ w(\tilde{T}) \leq w(\tstar) $.

    \vspace{2mm}
    \hspace{4mm} We first claim that $ \sigma(\tilde{T}) \leq (1+\varepsilon)B $. Note that $ 0\leq \frac{n-2}{\varepsilon B}\cdot\sigma(t) - \left\lfloor \frac{n-2}{\varepsilon B} \cdot \sigma(t) \right\rfloor \leq 1 $.  Therefore,
        \begin{equation*}
        \begin{split}
            \sigma(\tilde{T})
                = \sum_{t\in \tilde{T}} \sigma(t)
                &\leq \left( (n-2) + \sum_{t\in\tilde{T}} \left\lfloor \frac{n-2}{\varepsilon B}\cdot\sigma(t) \right\rfloor \right) \frac{\varepsilon B}{n-2} \\[2mm]
                &= \left( (n-2) + \sum_{t\in\tilde{T}} \tilde{\sigma}(t) \right) \frac{\varepsilon B}{n-2}
                \leq \left( (n-2) + \left\lfloor \frac{n-2}{\varepsilon} \right\rfloor \right) \frac{\varepsilon B}{n-2} \\[2mm]
                &\leq \left( (n-2) + \frac{n-2}{\varepsilon} \right) \frac{\varepsilon B}{n-2} = (1+\varepsilon)B.
        \end{split}
        \end{equation*}

    \hspace{4mm} We now claim that $ w(\tilde{T}) \leq w(\tstar) $. This follows if we instead prove that $\tilde{\sigma}(\tstar) \leq \tilde{B}$, because if these two triangulations satisfy the same quality constraint then by the minimality of $w(\tilde{T})$, it follows that $ w(\tilde{T}) \leq w(\tstar) $. This can be shown as follows.
    \begin{equation*}
    \begin{split}
        \tilde{\sigma}({\tstar})
        &= \sum_{t\in{\tstar}} \tilde{\sigma}(t)
        = \sum_{t\in{\tstar}} \left\lfloor \frac{n-2}{\varepsilon B} \cdot \sigma(t) \right\rfloor \\
        &\leq \left\lfloor \frac{n-2}{\varepsilon B} \cdot \sum_{t\in{\tstar}} \sigma(t) \right\rfloor
        \leq \left\lfloor \frac{n-2}{\varepsilon B} \cdot B \right\rfloor = \tilde{B},
    \end{split}
    \end{equation*}
as we desired.
\end{enumerate}
\end{proof}

\subsection{Putting Everything Together: Proof of \Cref{thm: dnttheorem}}\label{sec: putting-everything-together}

Recall that at the end of \Cref{sec: bct-introduction}, we proved that we can obtain $k$ distinct $\alpha$-optimal triangulations with diversity at least $ \betak $ of optimal, in time $ O(f_k(w_k,\sigma)k^2\log k) $, where $ f_k(w_k,\sigma) $ denotes the running time for computing the $k$-best $\alpha$-optimal $\BCT(w_k,\sigma)$.

(\ref{dnt: general}) Note that any allowed diagonal of $P$ can appear at most $k$ times in any set of $k$ triangulations of $P$.
Therefore, $w_k(T)$ is at most $nk$. Hence, by \Cref{thm: bct-algorithms}(\ref{itm: bct-algorithm-integer-measure}), $ f_k(w_k,\sigma) = O((nk)^2kn^3) = O(k^3n^5) $. Therefore, the running time of the DNT algorithm for $\alpha >1 $ is $ O(n^5k^5\log k) $.

(\ref{dnt: fptas}) Follow the similar lines in the preceding proof except to use \Cref{thm: bct-algorithms}(\ref{itm: bct-algorithm-fptas}).

%% file: 04_Conclusion.tex
\section{Conclusion}\label{sec: conclusion}
In this paper, we introduced and studied the Diverse and Nice Triangulations (DNT) problem, which seeks to find $k$ triangulations of a given simple polygon that maximize their diversity measured by symmetric differences, while ensuring each triangulation meets a predefined quality criterion. For the sum-DNT problem, we provided a polynomial-time approximation algorithm for the general case and presented a PTAS for certain special cases. For the min-DT problem, we highlighted potential computational hardness through a reduction from the Hamming code problem and showed a $(\tfrac12$)-approximation algorithm.

An intriguing direction for future research is to explore diversity measures based on \emph{flip distance}. While Aichholzer et al.~\cite{aichholzer2015flip} showed that computing the flip distance between two given triangulations is NP-complete, it is not clear if this result applies to computing $k=2$ most diverse solutions, a.k.a. the \emph{diameter} of the flip graph.
Another interesting direction would be to consider specific inputs, such as point sets without empty pentagons, for which the flip distance between triangulations \emph{can} be computed in polynomial time by a result of Eppstein~\cite{eppstein2007happy}.
In fact, Eppstein's methods can be used to define an \emph{earth mover's distance} between two triangulations, which can be computed in polynomial time and is a lower bound on the flip distance.
Finding diverse triangulations w.r.t. earth mover's distance would be a natural—and perhaps tractable—way to guarantee triangulations that are also diverse w.r.t. flip distance.

%% file: 05_DNT_Hardness.tex
\section{NP-Hardness of the Decision Versions of DNT}\label{sec: DDNT-Hardness}

In this section, we prove that the decision version of the DNT (DDNT) problem is NP-hard.
We restate the theorem for the reader's convenience. 

\ResDNTHardness*

Let $\#NT$ denote the problem of counting the number of nice triangulations of simple polygons. We prove \Cref{thm: DDNT-Hardness} by showing that $\#NT$ is $\#P$-hard and that, with at most $O(n)$ calls to a DDNT oracle, one can solve $\#NT$.

\begin{theorem}[\#P-hardness for \#NT]\label{thm: sharp-hardness-nt}
    There exists a simple polygon with $n$ vertices and some $\alpha'\in(1,\infty)$ such that no $\mathrm{poly}(n,\log k)$-time algorithm can count all $\alpha'$-optimal triangulations of $P$, unless $\pisnp$.
\end{theorem}
\begin{proof}
We reduce from a \#P-hard variant of the Subset-Sum problem (SS)~\cite{kleinbergtardos2005}\footnote{Although \cite{kleinbergtardos2005} proves the NP-hardness of SS, it is straightforward to extend this to \#P-hardness.}, denoted by \#SS. The \#SS problem asks, given a multiset $S$ of $n$ positive integers $v_i$ and a target value $V$, to compute the number of subsets $S' \subseteq S$ such that the sum of the elements in $S'$ equals $V$, \emph{i.e.}, $\sum_{v_i \in S'} v_i = V$.

\vspace{2mm}\noindent\textbf{Construction of hardness polygon.}
See \Cref{fig: sharp-p-Euclidean}. The polygon $P$ in the figure consists of $n$ \emph{kites}, where the $i$-th kite corresponds to the $i$-th item in $S$, all glued to a large right triangle.
For each kite, the vertical diagonal is shorter than the horizontal one, with the length of the shorter diagonal being $v_i$ and the longer diagonal being $2v_i$, so the difference between them is always $v_i$. Note that all diagonals of $P$, except those of the kites, are fixed and must form part of any triangulation of $P$. Moreover, exactly one diagonal of each kite can be included in a triangulation of $P$.
\begin{figure}[!ht]
    \centering
    \includegraphics[width=.8\linewidth]{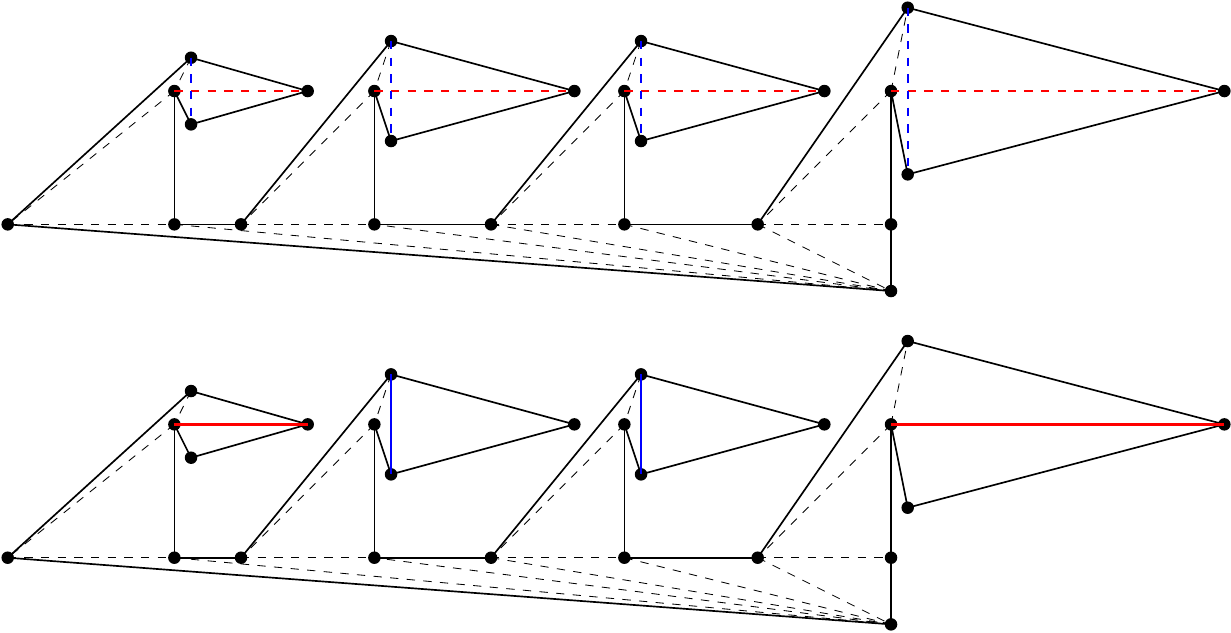}
    \caption{\small{An illustration for the reduction from \#SS to \#NT. \textbf{Above:} The hardness polygon $P$ constructed from an instance of \#SS given by ${v_1 = 1, v_2 = 4, v_3 = 4, v_4 = 6, V = 7}$.
    The dashed lines represent diagonals of $P$. All black diagonals always appear in any triangulation of $P$, and there are only two choices for each kite gadget: either {\color{red}red} (horizontal) or {\color{blue}blue} (vertical). Each red diagonal is twice as long as the blue one in the same kite. \textbf{Below:} Illustration for the encoding of the subset $\{ v_1, v_4 \}$. Note that the Euclidean length of the triangulation is $L + 1 + 4 = L + 7$, where $L$ is the MWT weight of $P$.}
    }
    \label{fig: sharp-p-Euclidean}
\end{figure}

\vspace{2mm}\noindent\textbf{The reduction.} Let $\sigma(T)$ denote the Euclidean length of a triangulation $T$, and let $L$ be the minimum Euclidean length of a triangulation of $P$, and let $h_i$ denote the horizontal diagonal of the $i$-th kite in $P$. Define a mapping $g$ that assigns each subset $S' \subseteq S$ to a triangulation $T$ of $P$ such that 1) $\sigma(T) = L + \sum_{v \in S'} v$, where $\sigma$ denotes the Euclidean length, and 2) $v_i \in S'$ if and only if $h_i \in T$.
From the construction, it is straightforward to verify that $g$ is a one-to-one correspondence between subsets of $S$ and triangulations of $P$. Moreover, only those subsets of $S$ whose elements sum to the target value $V$ are mapped to a triangulation $T$ of $P$ with Euclidean length $\sigma(T) = L + V$.

This proves that counting the number of triangulations with a given target Euclidean length is as hard as counting the number of solutions to a subset-sum problem; thus, counting the number of nice triangulations is \#P-hard.
\end{proof}

\vspace{2mm}We are now ready to prove \Cref{thm: DDNT-Hardness}.
Note that $P$ can have at most $ 4^n$ triangulations~\cite{DUTTON1986211}.
We show that running a DDNT oracle at most $O(n)$ times one can solve the \#NT problem.
Given an instance of \#NT, run the DDNT oracle for $ D = {n \choose 2} $ and $k\leq 4^n $ repeatedly until the oracle returns ``YES.'' By using the binary search technique, this can be done by running the oracle at most $ O (\log 4^n ) = O(n) $ times.
This completes proof.

%% file: 06_BCT_Hardness.tex
\section{NP-hardness of BCT}\label{sec: bct-hardness}
In this section we prove NP-hardness of the BCT problem with near-Delaunay measures. We first state the theorem followed by its proof sketch, and then provide detailed proof.

\ResBCTHardness*

For the definitions of the near-Delaunay measures, see \Cref{fig:delaunaymeasures-app}.

\begin{figure}[!h]
    \centering
    \begin{subfigure}[t]{0.24\textwidth}
        \centering
        \includegraphics[width=.9\textwidth]{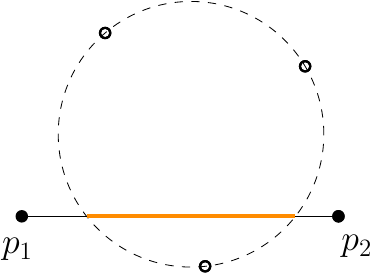}
        \subcaption{Shrunk Circle}
    \end{subfigure}
    \hfill
    \begin{subfigure}[t]{0.24\textwidth}
        \centering
        \includegraphics[width=.9\textwidth]{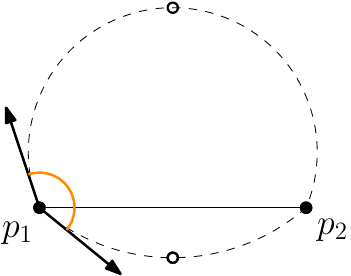}
        \subcaption{Lens}
    \end{subfigure}
    \hfill
    \begin{subfigure}[t]{0.24\textwidth}
        \centering
        \includegraphics[width=.9\textwidth]{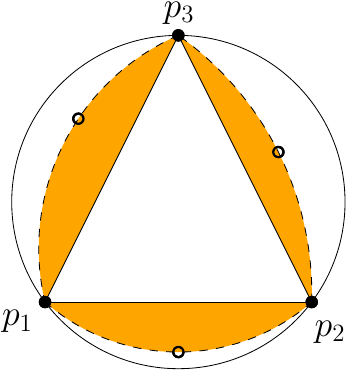}
        \subcaption{Triangular Lens}
    \end{subfigure}
    \hfill
    \begin{subfigure}[t]{0.24\textwidth}
        \centering
        \includegraphics[width=.9\textwidth]{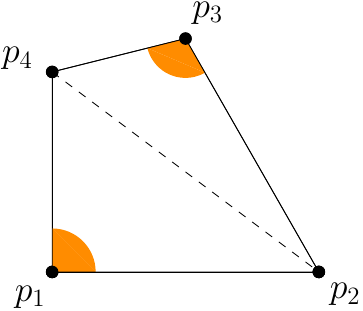}
        \subcaption{Opposing Angles}
    \end{subfigure}
    \caption{\small{The following measures, called near-Delaunay measures, were proposed by~\cite{van2021near,o2018open}. (a) Shrunk-circle measure ($\edgesh$), equals $\sum_{e \in T} \sigma(e)$, where $ \sigma(e) $ is the maximum fraction of the diagonal $e$ overlapped with the largest empty circle. (b) Lens-based measure($\edgelen$), equals $\sum_{e \in T} \sigma(e)$, where $\sigma(e) =\min\{1, \theta /\pi\} $ and $\theta$ is the angle formed by the tangent vectors of the two largest empty arcs on both sides. (c) Triangular-lens measure ($\trigdel$), equals $\sum_{t \in tr(T)} \sigma(t)$, where $\sigma(t)$ is the fraction of the area in the circle but outside the triangle that is covered by the shown lens. (d) Opposing angles measure ($\quaddel$), equals $\sum_{e \in T} \sigma(e)$, where $\sigma(e)=\max\{1,\theta(e)/\pi \}$ and $\theta(e)$ is the sum of opposing angles in the quadrilateral that has $e$ as a diagonal. Note that this measure is an example of a near-Delaunay measure that is not decomposable, since the opposing angles depend \emph{not} only on the diagonal $e$, but also on the triangles in $T$ adjacent to $e$.}}
    \label{fig:delaunaymeasures-app}
\end{figure}

Let $\uparrow$ denote maximization and $\downarrow$ denote minimization. For instance, $\BCT(\downarrow, \quaddel)$ denotes the problem of minimizing the Euclidean length of a triangulation of $P$, subject to $\quaddel \geq B$ for a given bound $B$.
We first generalize $\quaddel$ to $\quaddel_g$, where $g$ is any strictly increasing continuous function, then we prove hardness of $\BCT(\downarrow E,\quaddel_g)$ and $\BCT(\uparrow E,\quaddel_g)$, therefore hardness of $\BCT(\downarrow E, \quaddel)$ and $\BCT(\uparrow E, \quaddel)$ follow. Subsequently, we argue that there are some $ g $ with which  $\BCT(\downarrow E,\quaddel_g)$ or $\BCT(\uparrow E,\quaddel_g)$ can be reduced to its analogous $\edgelen$, $\edgesh$ and $\trigdel$ versions.

We show a reduction from the classical 0/1-Knapsack problem: $\{\{v_i,w_i\}_{i=1}^n$, $W\}$.
We create a polygon that is composed of \emph{kites} (See \Cref{fig: reduction-sketch}), each of which is linked to an isosceles trapezoid connected to a large right triangle, where the difference between the lengths of the two diagonals of the $i$-th kite is $v_i$ in the Knapsack problem.
If the left angle of a kite is small (resp., big), then opting for the longer (resp., shorter) diagonal of the kite results in a trade-off.
By setting the length of the shorter (resp., longer) diagonal of the $i$-th kite to $v_i$ (resp., $2v_i$), we make choosing the shorter (resp. longer) diagonal represent selecting (resp., not selecting) the $i$-th item.
Furthermore, we define a mapping between the weights of the Knapsack problem and the left angles of the kites so that the Delaunay trade-off caused by selecting a diagonal of the $i$-th kite represents weight trade-off in the Knapsack problem.

\begin{figure}[!t]
    \centering
    \includegraphics[width=.9\linewidth]{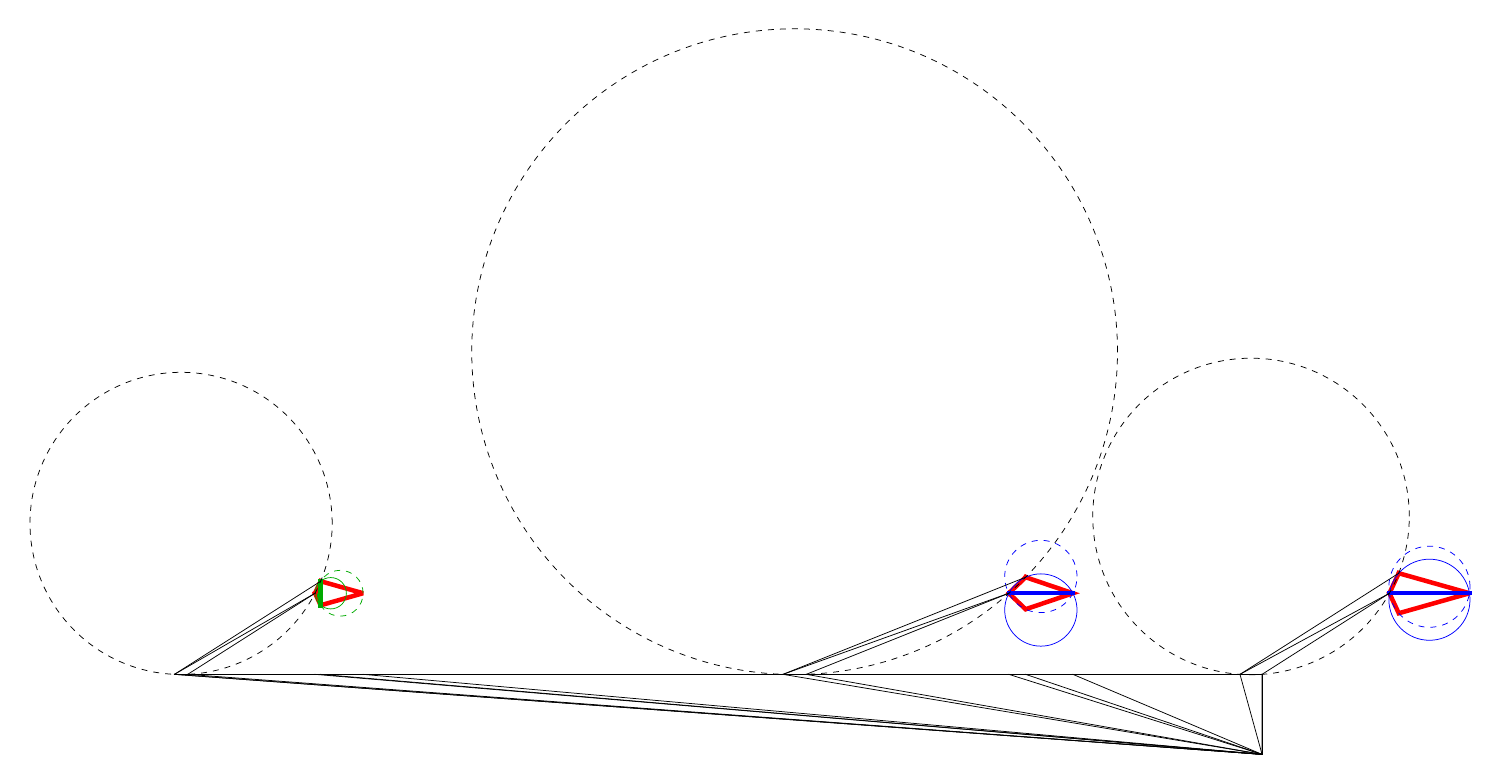}
    \vspace{5mm}
    \includegraphics[width=.9\linewidth]{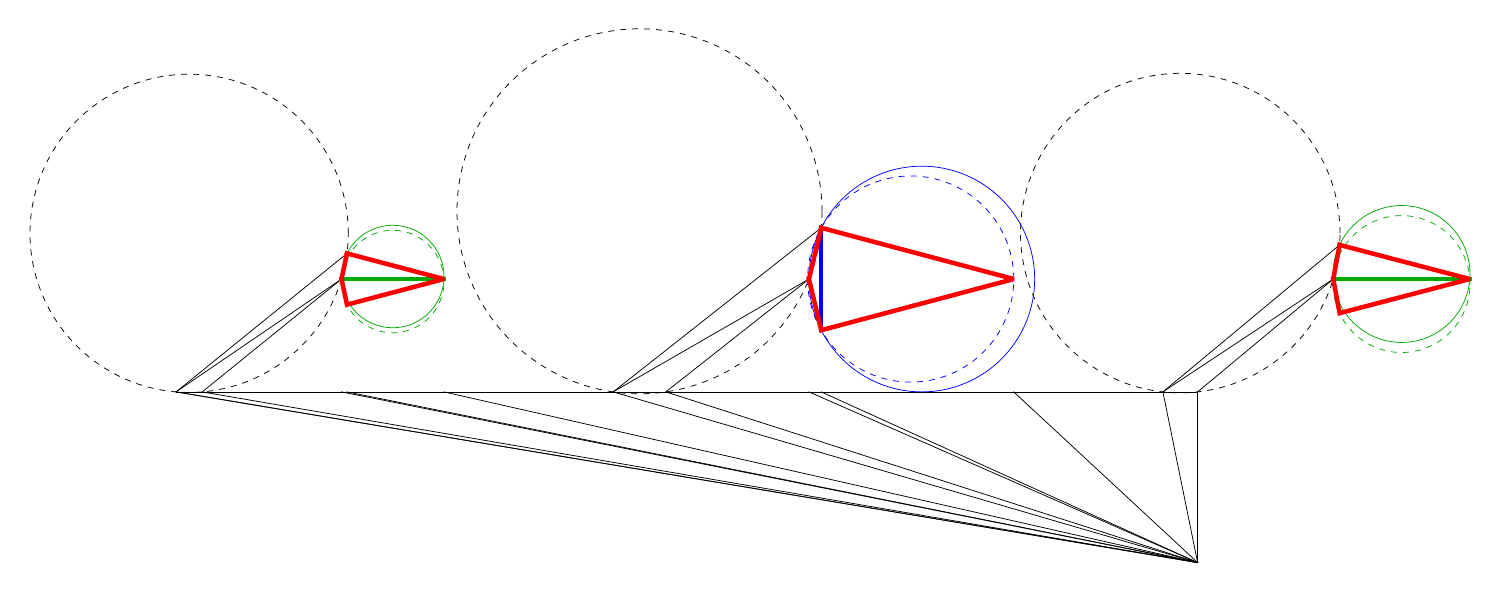}
    \captionsetup{width=.9\linewidth}
    \caption{Sketch for the reduction of knapsack to $\BCT(\uparrow{}E,\quaddel_g)$ (\emph{above}) and $\BCT(\downarrow E,\quaddel_g)$ (\emph{below}). The {\color{blue}blue} diagonals represent a choice of the associated item. Each {\color{blue}blue} circle is not empty, thus for each {\color{blue}blue} diagonal results in a Delaunay trade-off.
    }
    \label{fig: reduction-sketch}
\end{figure}

\begin{definition}[Generalized Opposing Angles Near-Delaunay Measure] \label{def: delaunay-soa-generalized}
    Let \( g:[\pi,2\pi] \rightarrow [0,1] \) be a \emph{strictly increasing continuous} function such that $g(\pi)=0$ and $g(2\pi)=1$. Given a quadrilateral \( \Box\,abcd \) and its diagonals \( \overline{ac} \) and \( \overline{bd} \), define \( \quaddel_g \) by
    \begin{align}
        {{\quaddel}_{g}}(\overline{ac}, abcd) = 
        \begin{cases}
            0, & \text{if} \quad (\angle b + \angle d) \in [0,\pi], \\
            g(\angle b + \angle d), & \text{if} \quad (\angle b + \angle d) \in [\pi,2\pi],
        \end{cases}
    \end{align}
    and similarly for \( \overline{bd} \).
    We may write \( \quaddel_g(\overline{ac}, \Box\,abcd) \) as \( \quaddel_g(\overline{ac}) \) if that does not cause any confusion.
\end{definition}

Before we go further, we provide two propositions which illustrate limitations of angles of item gadgets that are to be used in our reduction.

\begin{proposition} \label{prop: p2-angle-range}
    Let $\Box\,p_1p_2p_3p_4$ be an orthogonal quadrilateral such that both diagonals are lines of symmetry and $ \overline{p_2p_4} = 2 \cdot \overline{p_1p_3}$;
    see \Cref{fig: two-types-kites} for illustration. Then, opting the shorter (resp. longer) diagonal leads to a non-Delaunay triangulation if $ 3\pi/2 < \angle\,p_2  < \pi $ (resp. $ 2\arctan(1/2) < \angle\,p_2  < 3\pi/2 $).
\end{proposition}
\begin{proof}
    Let $p$ be the intersection of the two diagonals. Wlog, let $ \overline{p_1p} = 1 $. Then, $ \overline{p_2p_4} = 4 $. Put $ x = \overline{p_2p} $. Recall that $ \angle\,p_1 = \pi/2 $ \emph{if and only if} $ \overline{p_1p}^2 = \overline{p_2p} \cdot \overline{pp_4} $. Since $ x = 2-\sqrt{3} $ is a solution to $ 1^2 = x(4-x) $ and $ \cot(3\pi/2) = 2-\sqrt{3} $, we have proved the first case.

    To prove the second case, notice that we may assume that $ \overline{p_2p} \leq 2 $ by symmetry. When $ \overline{p_2p} = 2 $, $ \angle\,p_1p_2p = \arctan(1/2) $; we have proved the second case. 
\end{proof}

    \begin{figure}[H]
        \centering
        \includegraphics[width=\linewidth]{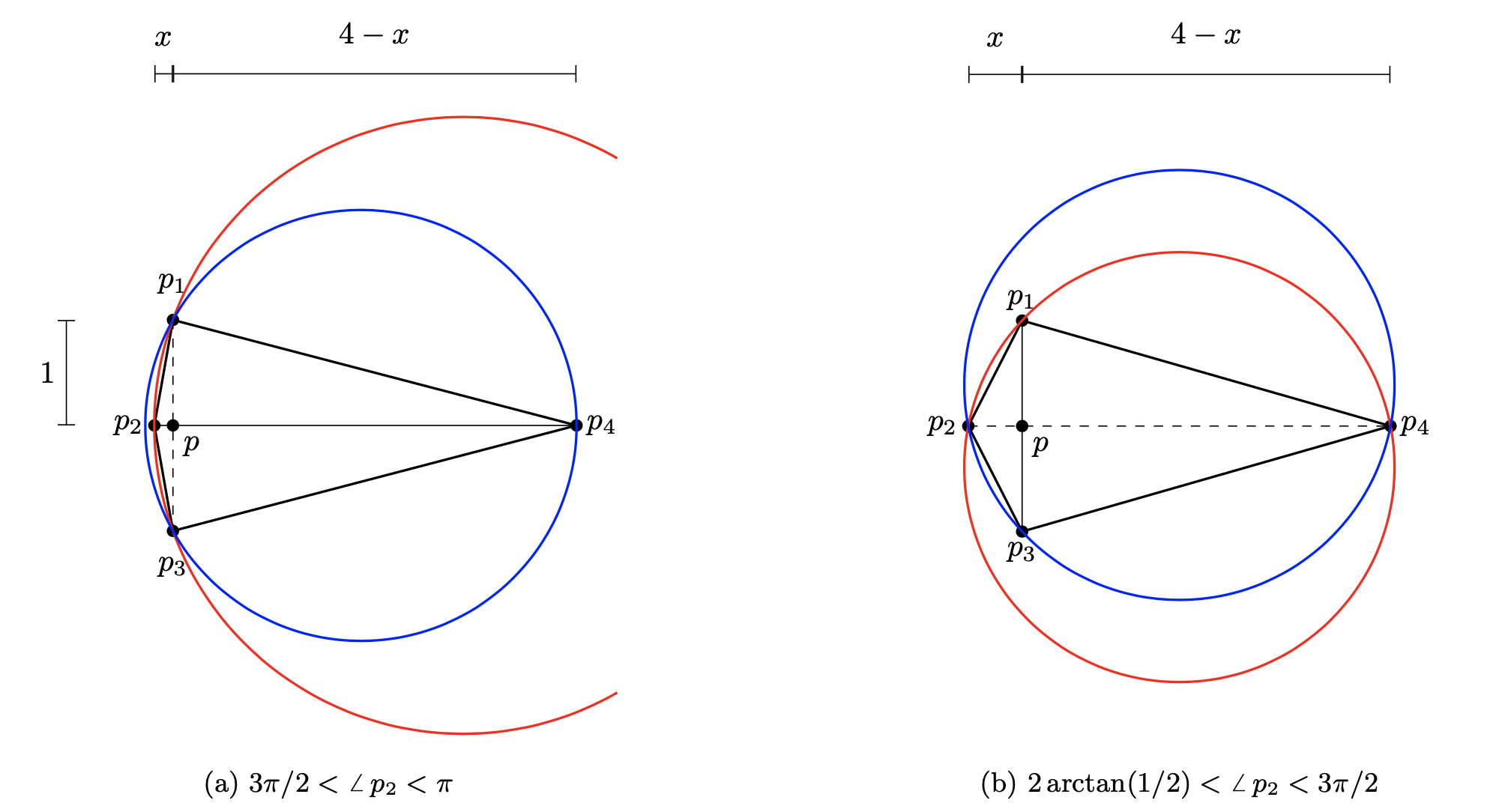}
        \caption{\textbf{Left:} $ 3\pi/2 < \angle\,p_2 < \pi $, and \textbf{Right:} $ 2\arctan(1/2) < \angle\,p_2 < 3\pi/2 $. Non-Delaunay triangulations of orthogonal quadrilaterals depending on the size of $ \angle\, p_2 $, where both diagonals are lines of symmetry and $\overline{p_2p_4} = 2\cdot\overline{p_1p_3}$. Opting for the shorter (resp. longer) diagonal leads to a non-Delaunay triangulation if $ 3\pi/2 < \angle\,p_2  < \pi $ (resp. $ \pi/6 < \angle\,p_2  < 3\pi/2 $).}
        \label{fig: two-types-kites}
    \end{figure}

\begin{proposition} \label{prop: p1-x-coord}
    Let $\Box\,p_1p_2p_3p_4$ be the same quadrilateral with the same constraints as given in \Cref{prop: p2-angle-range}. Let $ \beta = \angle\,p_1 $, and let $ x = \overline{p_2p} $, where $ 0 < x \leq 2 $. Then,
        \begin{equation}\label{eq: x}
            x =
            \begin{cases} 
                2 - \sqrt{3 + 4\cot(\beta)}, & \quad \text{if } \beta \in \left(\arctan(4),2\arctan(2)\right]\setminus \{\pi/2\}, \\
                2-\sqrt{3}, & \quad \text{if } \beta = \pi/2.
            \end{cases}
        \end{equation}
\end{proposition}
\begin{proof}
    The second case was already proven in \Cref{prop: p2-angle-range}. Put $ \beta_1 = \angle\,pp_1p_2 $ and $ \beta_2 = \angle\,pp_1p_4 $. Then,
        \[
            \tan \beta = \tan (\beta_1 + \beta_2) = \frac{x+(4-x)}{1-x(4-x)} = \frac{4}{x^2-4x+1},
        \]
    which yields $ x = 2 \pm \sqrt{3 + 4\cot(\beta)} $. Since $ \beta > \arctan(4) $, $ 0 < \sqrt{3 + 4\cot\beta} < 2 $. Equality on the right hand side holds when $ \beta = \arctan(-4/3)+\pi $. Since $ \beta $ is at its maximum when the quadrilateral is symmetric, $ \arctan(-4/3)+\pi = 2\arctan(2) $, hence \Cref{eq: x}.
\end{proof}

We are now ready to prove the hardness of BCT w.r.t $\BCT(\downarrow E,\quaddel_g)$ and $\BCT(\uparrow E,\quaddel_g)$.
\begin{lemma}\label{lem: hardness-quaddel}
    Let $g: [\pi,2\pi]\to [0,1]$ be a strictly increasing continuous function such that $g(\pi)=0$ and $g(2\pi)=1$. 
    Then $\BCT(\downarrow E,\quaddel_g)$ and $\BCT(\uparrow E,\quaddel_g)$ are NP-Hard.
    In particular, $\BCT(\downarrow E,\quaddel)$ and $\BCT(\uparrow E,\quaddel)$ are NP-hard.
\end{lemma}
\begin{proof} \textbf{(Hardness of $\BCT(\downarrow E,\quaddel_g)$)} Given a strictly increasing continuous function $g: [\pi,2\pi]\to [0,1]$ with $g(\pi)=0$ and $g(2\pi)=1$, assume without loss of generality that a decision version of Knapsack problem is defined with $ \{ v_i \in \mathbbm{Z}_{\geq 2} \}_{i=1}^n $, $ \{ w_i \in \mathbbm{Z}^+ \}_{i=1}^n $, $ V \in \mathbbm{Z}^+ $ and $ W \in \mathbbm{Z}^+ $, for which $ \max w_i < W $.
    
Using $v_i$ and $w_i$, we will construct in polynomial time a simple $(7n-2)$-gon, $P$, which consists of $n$ \emph{kites}, each kite of which is paired with an isosceles trapezoid glued to a large right-triangular $(5n-2)$-gon illustrated in \Cref{fig: rhombi-gadgets}. Furthermore, $P$ will be Delaunay-triangulable except the kites so that, on $P$, $ \quaddel_g $ will have \emph{strict monotone} relationship with all other near-Delaunay measures we are interested in; thus, proving hardness of $ \BCT(\downarrow E, \quaddel_g) $ indeed proves hardness of $ \BCT(\downarrow E, \sigma) $, where $
sigma$ is any near-Delaunay measure.

    \begin{figure}[H]
        \centering
        \includegraphics[width=\linewidth]{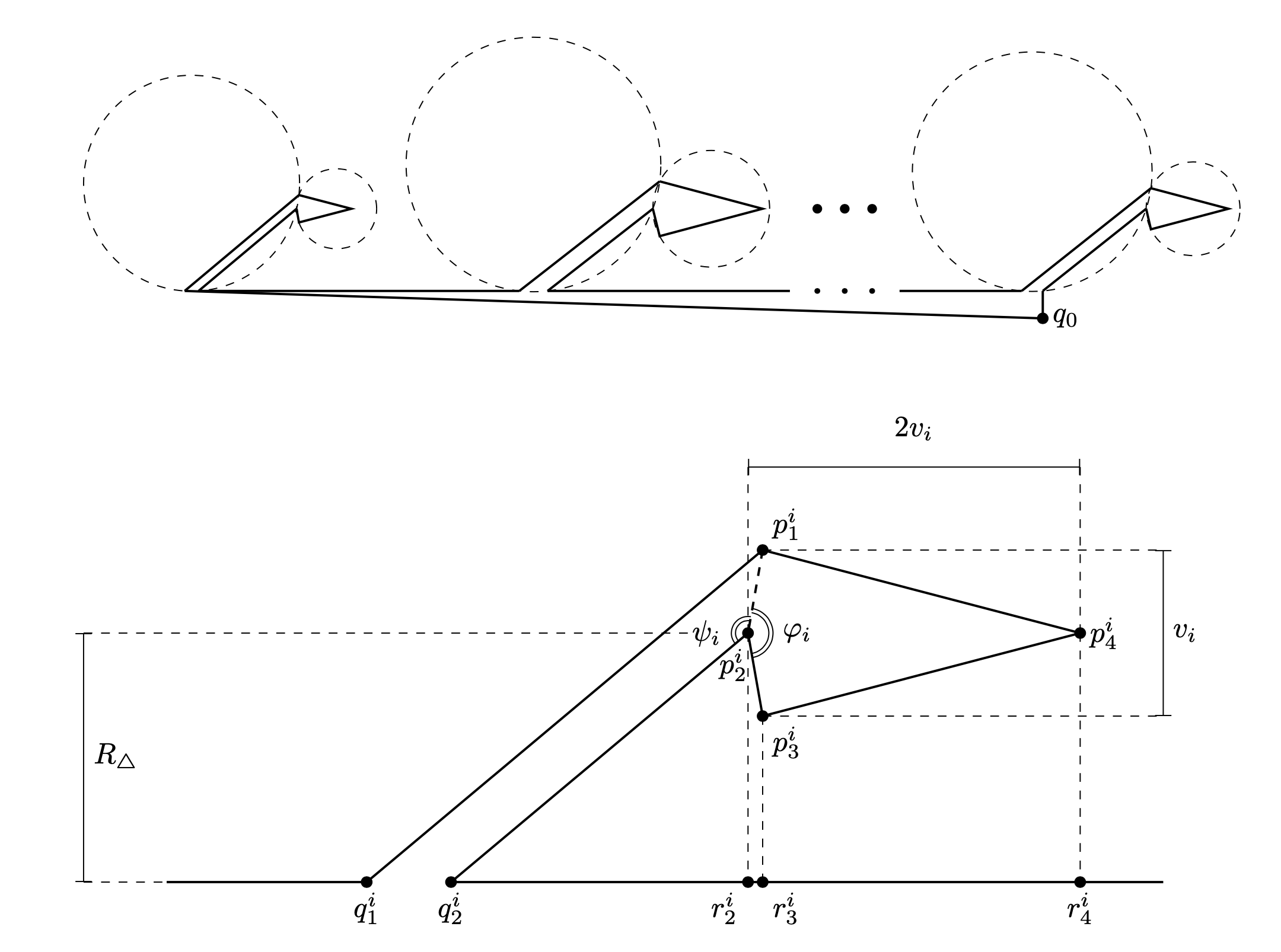}
        \caption{\textbf{(Above)} All gadgets are connected to a bus, which is shaped like a large right triangular $(5n-2)$-gon. Each item gadget is placed at a sufficient distance from both the preceding gadget and the upper border of the bus. This arrangement ensures that the circumcircle of any triangle from any gadget does not encompass any points belonging to other gadgets or the bus. Note that the bus has a unique triangulation.
        \textbf{(Below)} The $i$-th gadget that represents the $i$-th item of the given Knapsack problem. The \emph{kite}, $ \Box\, p_1p_2p_3p_4 $, is connected to an isosceles trapezoid. The trapezoid is long enough so that the largest circumcircle from the kite does not touch the upper border of the bus. $r^i_2$, $r^i_3$ and $r^i_4$ are perpendicular foot dropped from $p^i_1$, $p^i_2$ and $p^i_3$, respectively, so no circumcircles of triangles of the triangulation of the bus do not encompass any points of the kites.}
        \label{fig: rhombi-gadgets}
    \end{figure}

    For the given items, we first create orthogonal quadrilaterals, call the \emph{kites}, that were used in \Cref{prop: p2-angle-range}. Define $ \theta_i := \angle\,p^i_2 + \angle\,p^i_4 $. Note that we want $ g(\theta_i) $ to represent $ w_i $. If we find some constant $c$ such that $ g(\theta_i) = c \cdot w_i $, then we may measure $ \angle\,p^i_1 $ as $ \frac{2\pi-g^{-1}(c \cdot w_i)}{2} $. Then, we can find the $x$-coordinate of $ p^i_1 $, hence the whole figure of kite $i$. Set
        \begin{equation}\label{eq: c-val}
            c := \frac{g\left( \pi + 2\cot^{-1}(4) \right)}{W}
        \end{equation}
    so that $ 0 < c \cdot w_i < g\left( \pi + 2 \cot^{-1} (4) \right) < 1 $ for any $ i \in [n] $; $ g^{-1}(c \cdot w_i) $ is well-defined.
    Given $i$, put $ p^i_2 = (x^i_2,y^i_2) $ and find $ p^i_1 $, $ p^i_3 $ and $ p^i_4 $ such that
        \begin{equation}
            \overline{p^i_1p^i_3} = v_i, \quad
            \overline{p^i_2p^i_4} = 2v_i, \quad
            \overline{p^i_1p^i_3} \perp \overline{p^i_2p^i_4}, \quad
            \overline{p^i_1p^i_2} = \overline{p^i_2p^i_3}, \quad
            \angle\, p^i_1 = \frac{2\pi - g^{-1}(c \cdot w_i)}{2}
        \end{equation}
    by using \Cref{prop: p1-x-coord}.
    Let $x^i_1$ be the $x$-coordinate of $ p^i_1 $. Then,
        \begin{equation}
            \varphi_i := \angle\,p^i_2 = 2\arctan\left(\frac{v_i/2}{x^i_1-x^i_2}\right).
        \end{equation}
        
    We now determine the relative positions of $ q^i_1 $ and $ q^i_2 $ to $ p^i_2 $. Let $ R^i_1 $ be the largest circumradius of the triangles belonging to $ \Box\,p^i_1p^i_2p^i_3p^i_4 $.
    Let $ R_{\triangle} := \max_i R^i_1 $. Let $ \psi_i = \angle\,p^i_1p^i_2q^i_2 $. Then,  for $ \Box\,p^i_2q^i_2q^i_1p^i_1 $ to be an isosceles trapezoid, $ \angle\,p^i_1p^i_2q^i_2 = \angle\,p^i_2q^i_2q^i_1 = \psi_i = \pi - \theta_i/4 $.
    Also, we want each kite distant enough from the bottom right triangle gadget so that no circumcircle from the kite encompasses any point of the right triangle gadget.
    Define,
        \[
            q^i_2 := p^i_2 + \left(R_{\triangle} \cdot \cot \psi_i, -R_{\triangle} \right) \qquad \text{and} \qquad q^i_1 := q^i_2 - \left( \overline{p^i_1p^i_2}, 0 \right).
        \]
    Let us now determine $p^1_2$ and $p^i_2$ so that there is enough space between every pair of two consecutive gadgets. Denote the circumradius of $ \Box\,p^i_1p^i_2q^i_2q^i_1 $ by $ R^i_{\Box} $.
    Define
        \[
            p^1_2 := (0,0) \quad p^i_2 := p^{i-1}_4 + (2 R^i_{\Box} ,0).
        \]
    Since the circumcenter of the $i$-th trapezoid is above $ p^i_2 $, the circumcircle does not encompass $ p^{i-1}_4 $. 
    
    Define $ q_0 := q^n_2 + (0,-1) $. Note that the circumcircle of $\triangle q_0q^{i+1}_1q^i_2$ for some $i$ might encompass points of the $i$-th kite gadgets. To avoid this scenario, we put some extra points along the upper border of the right triangle gadgets. Define $ r^i_j $ by the projection of $ p^i_j $ onto $\overline{q^i_2q^{i+1}_1}$ for $ i\in[n-1] $ and $j\in\{2,3,4\}$, completing construction of $P$. Now, the following claim is obvious.
    
    \begin{claim} \label{claim: P-delaunay}
        The only triangulation of the right triangle gadget is locally Delaunay in $P$. In fact, there is a Delaunay triangulation of $P$.
    \end{claim}

    Define a decision version of $\BCT(\downarrow E, \quaddel_g)$ as a problem that asks that given $V'>0$ and $b\geq0$ if there exists a triangulation $T$ of $P$ such that $E(T) \leq V'$ and $\quaddel_g(T) \leq b$. Denote this problem by $\BCT((\downarrow E,\leq,V'),(\quaddel_g,\leq,b))$. We now prove that a YES-instance of the Knapsack problem is a YES-instance of $\BCT(\downarrow E, \quaddel_g)$ and vice versa.
    \begin{claim}\label{claim:yes-instance-knapsack-to-bct}
        Let $L_{\max}$ denote the maximum Euclidean weight of triangulations of $P$, and let $ c = g\left( \pi + 2\cot^{-1}(4) \right)\big/W $ as in \Cref{eq: c-val}. Then, a YES-instance of $ \textsc{Knapsack}(\{v_i\}, \{w_i\}, V, W) $ is a YES-instance of $ \BCT((\downarrow E,\leq,L_{\max}-V),(\quaddel_g,\leq,c\cdot W)) $.
    \end{claim}
    \begin{proof}
        Let a binary sequence $ \{{z^{\star}}_i\} $ be a YES-instance of the Knapsack problem. When triangulating $P$, select $ p^i_1p^i_3 $ if $ {z^{\star}}_i = 1 $ and $ p^i_2p^i_4 $ if $ {z^{\star}}_i = 0 $. Select any remaining diagonals to complete a triangulation of $P$, and denote by ${t^{\star}}$ the resulting triangulation. Since $ \overline{p^i_1p^i_3} - \overline{p^i_2p^i_4} = -v_i $ and $ \sum_i {z^{\star}} v_i \geq {V} $ , we have $E({t^{\star}}) = L_{\max} - \sum_i {z^{\star}} v_i \leq L_{\max} - V$.

        Also, by the first claim above,
        % \Cref{claim: P-delaunay},
            \begin{equation*}
            \begin{split}
                \quaddel{}_g({t^{\star}})
                    &= \sum_i \left({z^{\star}}_i\cdot\quaddel{}_g(\overline{p^i_1p^i_3},\Box\,p^i_1p^i_2p^i_3p^i_4)\right)
                    = \sum_i \Big({z^{\star}}_i \cdot g(\theta_i)\Big) \\
                    &= \sum_i \Big( {z^{\star}}_i \cdot g(g^{-1}(c \cdot w_i)) \Big)
                    = c \cdot \sum_i {z^{\star}} \cdot w_i
                    \leq c \cdot W,
            \end{split}
            \end{equation*}
        as we expected.
    \end{proof}

    \begin{claim}\label{claim:yes-instance-bct-to-knapsack}
        A YES-instance of $ \BCT((\downarrow E,\leq,L_{\max}-V),(\quaddel_g,\leq,c\cdot{W})) $ is a YES-instance of $ \textsc{Knapsack}(\{v_i\}, \{w_i\}, {V}, {W}) $.
    \end{claim}
    \begin{proof}
        Let ${t^{\star}}$ be a YES-instance of $ \BCT((\downarrow E,\leq,L_{\max}-V),(\quaddel_g,\leq,c\cdot{W})) $. Define $ {z^{\star}}_i $ by $ {z^{\star}}_i = 1 $ if $\overline{p^i_1p^i_3}\in{t^{\star}}$, and 0, otherwise. Since $ E({t^{\star}}) = L_{\max} - \sum_i {z^{\star}}_iv_i \leq L_{\max} - V $,
        it immediately follows that $ \sum_i {z^{\star}}_iv_i \geq {V} $. Also, since
            \begin{equation*}
                \quaddel{}_g({t^{\star}})
                    = \sum_i \left({z^{\star}}_i\cdot\quaddel{}_g(\overline{p^i_1p^i_3},\Box\,p^i_1p^i_2p^i_3p^i_4)\right)
                    = c \cdot \sum_i {z^{\star}}\cdot  w_i
                    \leq c \cdot {W},
            \end{equation*}
        we have that $ \sum_i {z^{\star}}_i\cdot w_i \leq {W} $, showing that our claim holds.
    \end{proof}

    \noindent \textbf{(Hardness of $\BCT(\downarrow E,\edgelen)$, $\BCT(\downarrow E,\edgesh)$ and $\BCT(\downarrow E,\trigdel)$)} Proving the hardness for $\edgelen$, $\edgesh$ and $\trigdel$ now follows since we have shown that $\BCT(\downarrow E,\quaddel_g)$ is hard for any $g$. As illustrated in \Cref{fig: monotone-relationship}, with $\overline{p^i_1p^i_3}$ and $\overline{p^i_2p^i_4}$ held fixed, as $\varphi_i$ increases the fractional overlap, the lens angle and the fractional triangular lens area decrease.
    Therefore, the larger $\varphi_i$ the less Delaunay opting $\overline{p^i_1p^i_3}$ leads to. Therefore, $ \BCT(\downarrow E, \sigma) $, where $ \sigma $ is any of $\edgelen$, $\edgesh$ and $\trigdel$, can solve $ \BCT(\downarrow E, \quaddel_g) $ such that $ \quaddel_g (\theta_e) = \sigma(e) $ or $ \quaddel_g (\theta_e) = \sigma(t) $, where $ e \in t \in T \in \mathcal{T} $. Hence, $\BCT(\downarrow E,\edgelen)$, $\BCT(\downarrow E,\edgesh)$ and $\BCT(\downarrow E,\trigdel)$ are all \emph{NP-hard}. \\

    \noindent \textbf{(Hardness of $\BCT(\uparrow E,\quaddel_g)$)} As composed to the $ \min $ case, we want opting the longer diagonal of a kite to lead to a locally non-Delaunay triangulation. I.e., we want
        \[
            \theta_i := g^{-1}(c \cdot w_i) = \angle\,p_1 + \angle\,p_3.
        \]
    By \Cref{prop: p2-angle-range}, $ \pi < \angle\,p_1 + \angle\,p_3 < 4\arctan(2) $. Thus, for $ g^{-1}(c \cdot w_i) $ to be well-defined, set
        \[
            c := \frac{g(4\arctan(2))}{{W}} \quad \text{and} \quad \angle\,p^i_1 = \frac{g^{-1}(c \cdot w_i)}{2},
        \]
    so that
        \[
            0 < c \cdot w_i < g(4\arctan(2)) \quad \text{and} \quad \theta_i = \angle\,p^i_1 + \angle\,p^i_3.
        \]
    Now, the rest of the proof follows analogously.
\end{proof}

\newpage
\begin{figure}[H]
    \centering
    \begin{subfigure}{\textwidth}
        \centering
        \includegraphics[width=0.2\textwidth]{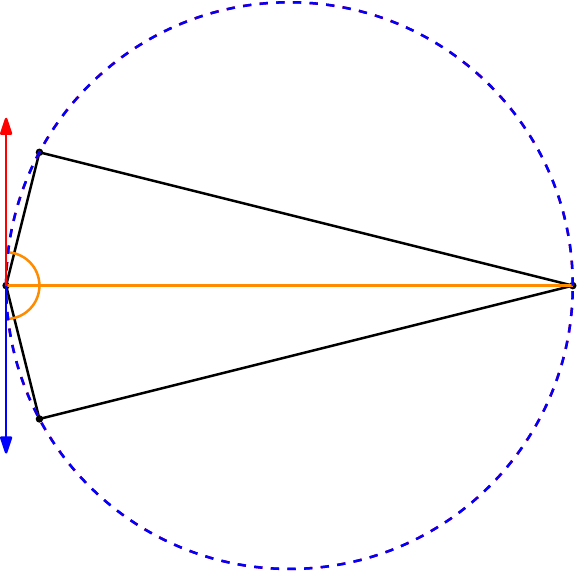}\hspace{2cm}
        \includegraphics[width=0.2\textwidth]{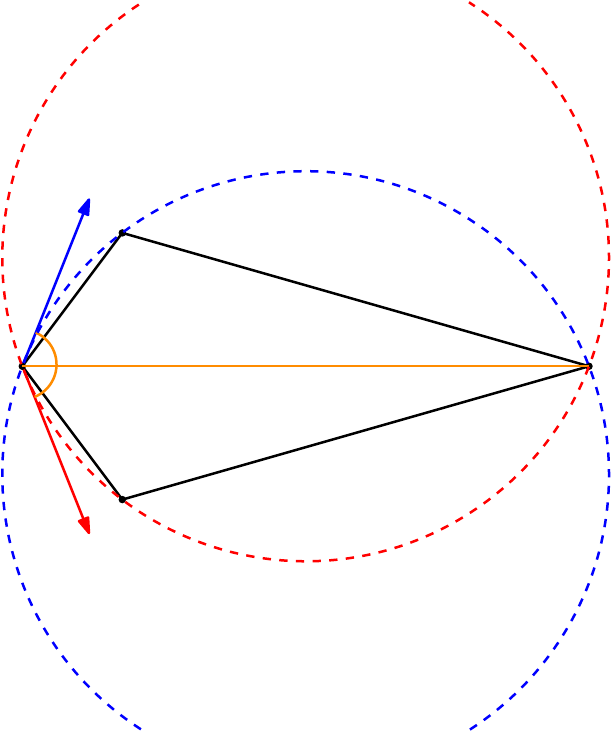}\hspace{2cm}
        \includegraphics[width=0.2\textwidth]{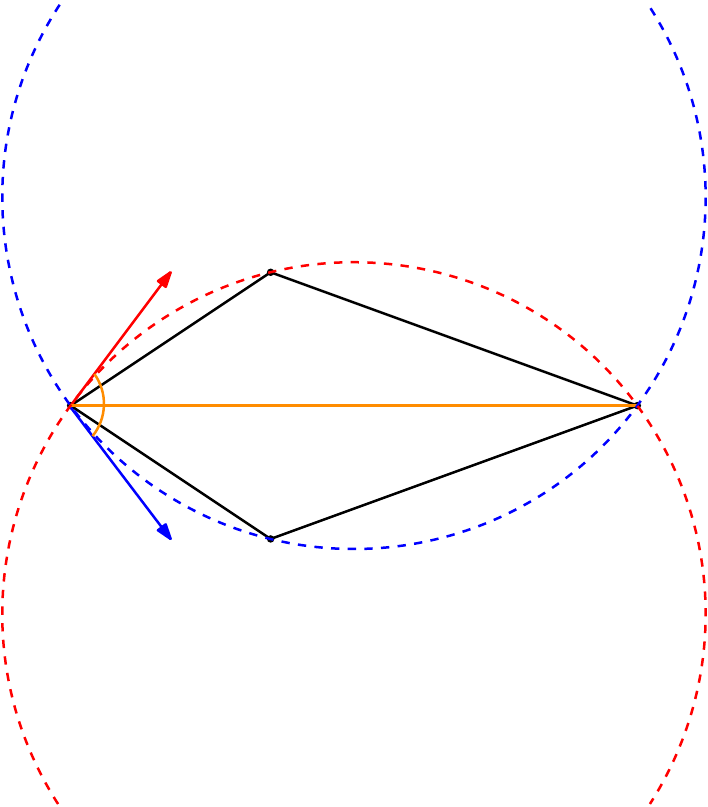}
        \label{fig:edgelenhard-row}
    \end{subfigure}
    \vspace{0.5cm}

    \begin{subfigure}{\textwidth}
        \centering
        \includegraphics[width=0.2\textwidth]{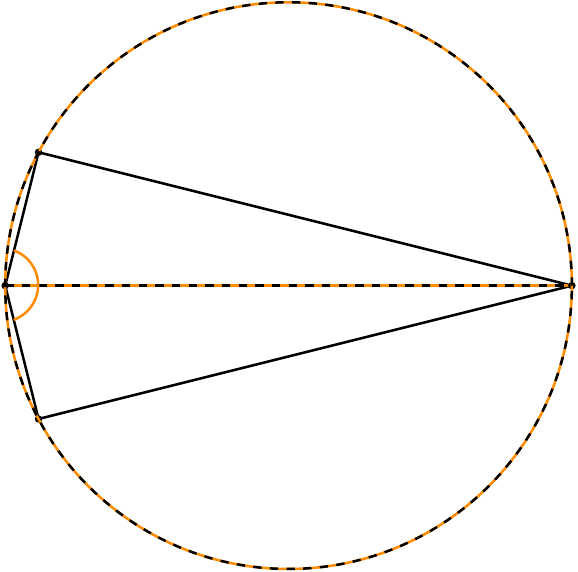}\hspace{2cm}
        \includegraphics[width=0.2\textwidth]{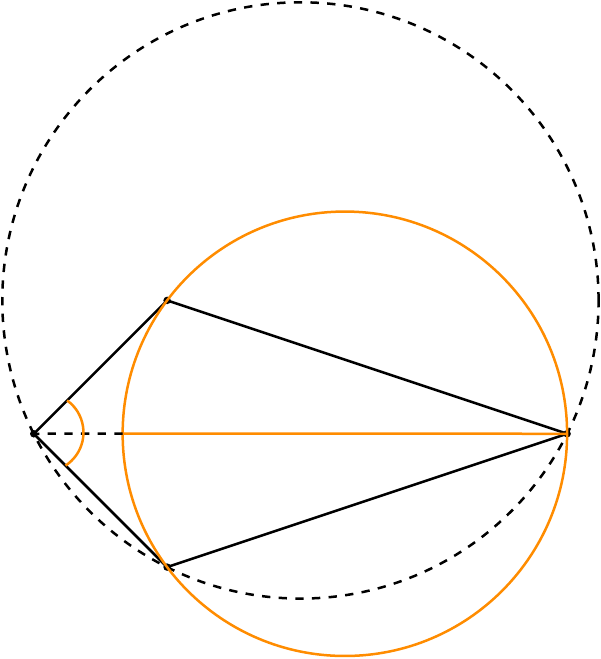}\hspace{2cm}
        \includegraphics[width=0.2\textwidth]{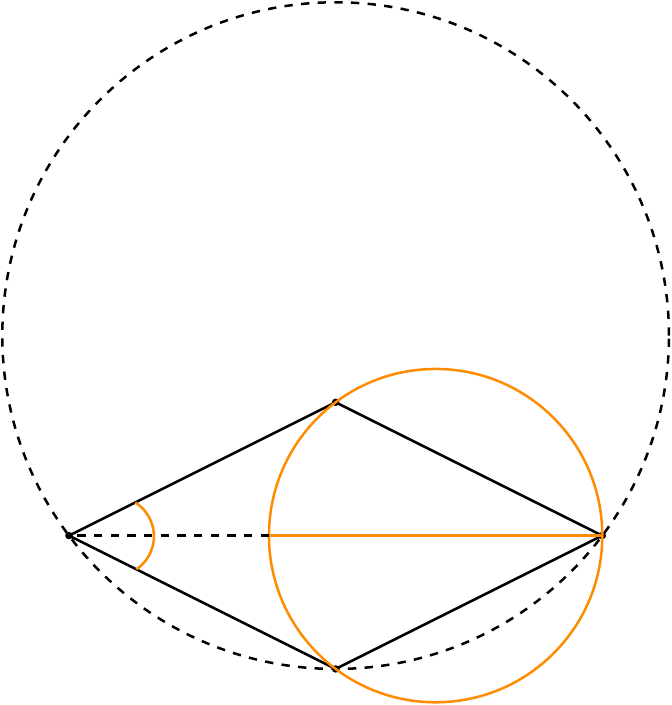}
        \label{fig:edgeshhard-row}
    \end{subfigure}
    \vspace{0.5cm}

    \begin{subfigure}{\textwidth}
        \centering
        \includegraphics[width=0.2\textwidth]{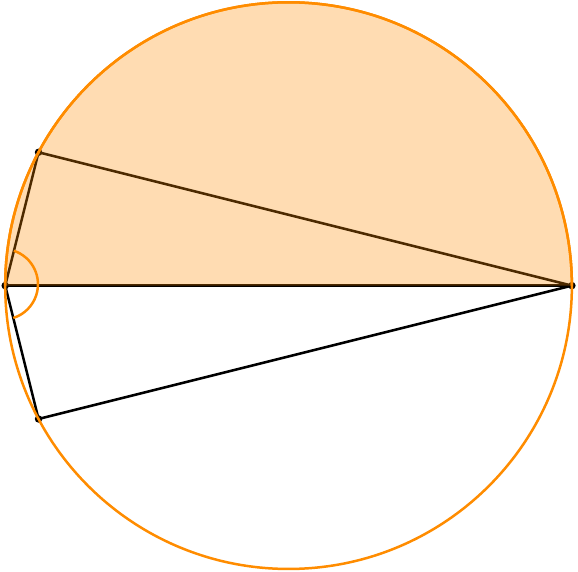}\hspace{2cm}
        \includegraphics[width=0.2\textwidth]{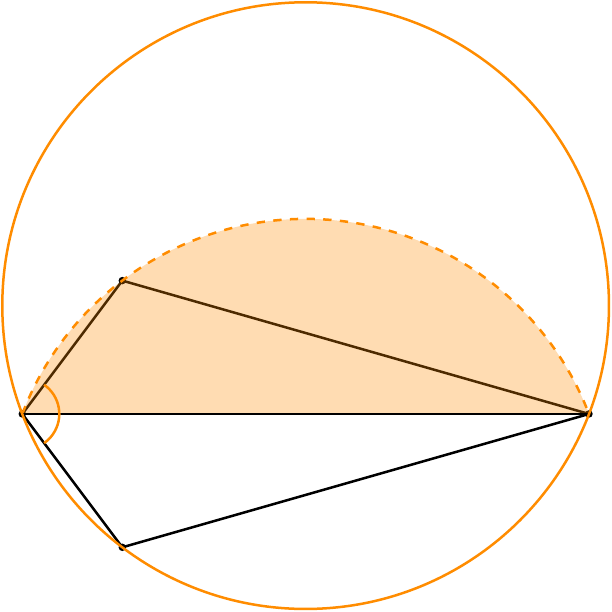}\hspace{2cm}
        \includegraphics[width=0.2\textwidth]{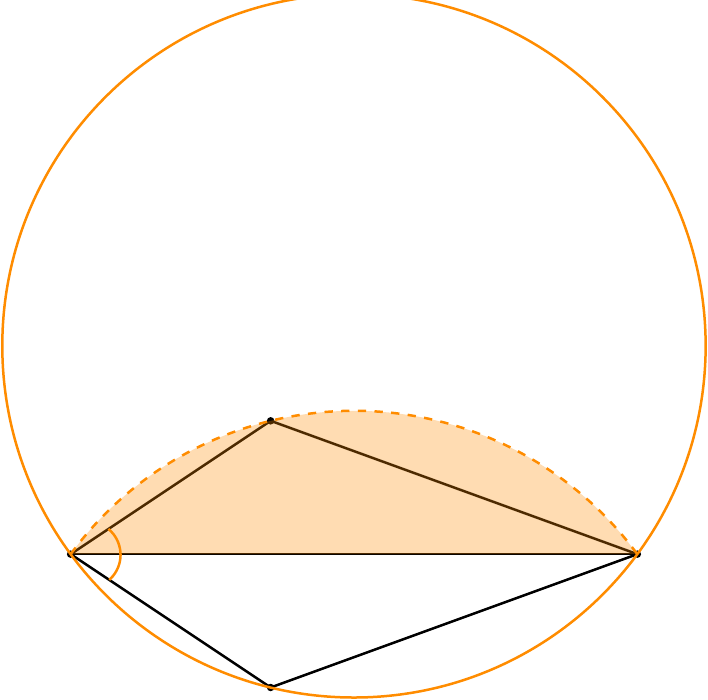}
        \label{fig:trigdelhard-row}
    \end{subfigure}

    \caption{Illustrations of the monotone relationships between $\quaddel_g$ and the other Delaunay measures—$\edgelen$, $\edgesh$, and $\trigdel$—on our hardness polygon. This illustration depicts that NP-Hardness of $\BCT(\uparrow E,\quaddel_g)$ implies that NP-Hardness of $\BCT(\uparrow E,\edgelen)$, $\BCT(\uparrow E,\edgesh)$ and $\BCT(\uparrow E,\trigdel)$.
    Due to space constraints, it is challenging to demonstrate the monotone relationships for decreasing ($\downarrow$) cases here; however, analogous monotone relationships for these cases can also be readily identified.}
    \label{fig: monotone-relationship}
\end{figure}

%% file: 07_BCT_details.tex
\section{Proof of \Cref{thm: bct-algorithms} (Algorithms for BCT)}\label{sec: bct-details}

We now present the full proof of \Cref{thm: bct-algorithms}, beginning with a restatement of the theorem for the reader’s convenience.

\ResBCTAlgorithms*
\begin{proof}
\begin{enumerate}
    \item\textbf{(Integer-valued measure)}
    
    \textbf{Case 1: $w(\cdot)$ is integer-valued and $w(T)\in[0,W]$.} This case has already been proven, so we omit it here.

    \vspace{2mm}\noindent\textbf{Case 2: $\sigma(\cdot)$ is integer-valued.} We provide proof for the case that both $w(\cdot)$ and $\sigma(\cdot)$ are triangle-decomposable since the case that one of the measure is edge-decomposable can be handled analogously.
    
    \hspace{4mm} Assume that $k=1$ and both $w(\cdot)$ and $\sigma(\cdot)$ are additive triangle-decomposable measures.
    Let $OPT(B',i,j)$ denote the minimum weight of triangulation of $P[i:j]$ whose quality is at most $B'$:
\[ OPT(B',i,j) = \min\{ w(T) \mid \sigma(T) \leq B', T\in\cT[i:j] \}. \]
By convention, we define $ OPT(B',i,j) = w(\triangle(i,i+1,i+2)) $ when $j=i+2$.
When $j>i+2$, then $ OPT(B',i,j) $ can be found by using the following recurrence relation:
\begin{equation*}
    OPT(B',i,j) = \min\limits_{\substack{m \in [i+2, j-2] \\ B'_1,B'_2 \in [0, B'-\sigma(\triangle imj)], \\ B'_1 + B'_2 = B'-\sigma(\triangle imj)}} OPT(B'_1,i,m) + w(\triangle imj) + OPT(B'_2, m,j)
\end{equation*}

When any of $\overline{im}$, $\overline{mj}$ and $\overline{mj}$ is not an allowed diagonal of $P$, or $B'<\sigma(\triangle imj)$, we set $OPT(B',i,j) = \infty$.
Then, it is easy to check that $ OPT(B,1,n) $ can be computed in time $ O((B+1)^2 n^3) $.

\hspace{4mm} When $k>1$, as in Case 1, let $OPT(B',i,j)$ denote the sorted tuple (in non-decreasing order) of the $k$ smallest $w(T)$ among $T\in\cT[i:j]$ with $ \sigma(T) \leq B' $.
The rest of the proof is analogous to Case 1.

\hspace{4mm} We end the proof by noting similar dynamic programs can be easily constructed when $w(\cdot)$ and $\sigma(\cdot)$ are not additive, but are decomposable with $\min$ or $\max$.
For example, let $w(\cdot)$ and $\sigma(\cdot)$ are both edge-decomposable with respect to $\min$, and $w(\cdot)$ ranges over $[0,W]$, let $ OPT(W',i,j) $ denote the best quality of a triangulation of $P[i:j]$, and define $ OPT(W',i,j)=\infty $ when $ j=i+1 $.
When $j>i+1$, then $OPT(W',i,j)$ can be computed by using the following recurrence relation:
\begin{equation*}
    OPT(W',i,j) = \min_{m, W'_1,W'_2} \Big\{ \min \{OPT(W'_1,i,m),\sigma(\triangle imj), OPT(W'_2, m,j)\} \Big\},
\end{equation*}
where the first minimum in the equation above is taken over all $m \in [i+1, j-1] $ and $ W'_1,W'_2 \in [0, W'-(\triangle imj)] $ such that $ W' = \min \{W'_1,W'_2,w(\triangle imj)\} $.

\hspace{4mm} The other cases can be similarly handled, we omit the details.

    \item \textbf{(FPTAS)} Assume that $ w(\cdot) $, $ \sigma(\cdot) $, and a budget $B\geq0$ are given. For brevity, we assume that both $ w(\cdot) $ and $ \sigma(\cdot) $ are additive and triangle-decomposable. The other cases can be handled similarly.
    
    \hspace{4mm} The main idea of the algorithm is to scale down the quality measures of the allowed triangles in the polygon, as well as the given budget, to integers of size at most $\mathrm{poly}(n, \varepsilon)$.  
    We then run the BCT algorithm described above in Case 2, over these scaled weights and budget. Since the running time of the algorithm in Case 2 is $O((B+1)^2 n^3)$, the BCT problem can be solved in cubic time when $B=0$. Thus, we focus ourselves on the case where $ B > 0$. 

    Assume that $ B > 0 $.
    Given $ \varepsilon > 0 $, define
        \begin{equation}
            \tilde{\sigma}(t) := \left\lfloor \frac{n-2}{\varepsilon B} \cdot \sigma(t) \right\rfloor, \qquad \tilde{B} := \left\lfloor \frac{n-2}{\varepsilon} \right\rfloor,
        \end{equation}
     where $t$ is any allowed triangle of $P$.
     Then we solve $ \BCT(w, \tilde{\sigma}) $ with bound $\tilde{B}$, i.e.,
     \begin{equation*}\label{eq: fptas-constraint}
         \argmin\{w(T):T \in \mathcal{T}\text{ and }\tilde{\sigma}(T) \,\,\leq\,\, \tilde{B}\}
     \end{equation*}
    which can be done in time $O(\tilde{B}^2 \cdot n^3) = O(\varepsilon^{-2} \cdot n^5)$ by using the BCT algorithm in Case 2.

    \hspace{4mm} Let $ \tilde{T} $ be an output triangulation, and $\tstar$ be an optimal triangulation.
    We now prove that the $\tilde{T}$ and $\tstar$ satisfy the desired conditions, i.e.,  $ \sigma(\tilde{T}) \leq (1+\varepsilon)B $ and $ w(\tilde{T}) \leq w(\tstar) $.

    \vspace{2mm}
    \hspace{4mm} We first claim that $ \sigma(\tilde{T}) \leq (1+\varepsilon)B $. Note that $ 0\leq \frac{n-2}{\varepsilon B}\cdot\sigma(t) - \left\lfloor \frac{n-2}{\varepsilon B} \cdot \sigma(t) \right\rfloor \leq 1 $.  Therefore,
        \begin{equation*}
        \begin{split}
            \sigma(\tilde{T})
                = \sum_{t\in \tilde{T}} \sigma(t)
                &\leq \left( (n-2) + \sum_{t\in\tilde{T}} \left\lfloor \frac{n-2}{\varepsilon B}\cdot\sigma(t) \right\rfloor \right) \frac{\varepsilon B}{n-2} \\[2mm]
                &= \left( (n-2) + \sum_{t\in\tilde{T}} \tilde{\sigma}(t) \right) \frac{\varepsilon B}{n-2}
                \leq \left( (n-2) + \left\lfloor \frac{n-2}{\varepsilon} \right\rfloor \right) \frac{\varepsilon B}{n-2} \\[2mm]
                &\leq \left( (n-2) + \frac{n-2}{\varepsilon} \right) \frac{\varepsilon B}{n-2} = (1+\varepsilon)B.
        \end{split}
        \end{equation*}

    \hspace{4mm} We now claim that $ w(\tilde{T}) \leq w(\tstar) $. This follows if we instead prove that $\tilde{\sigma}(\tstar) \leq \tilde{B}$, because if these two triangulations satisfy the same quality constraint then by the minimality of $w(\tilde{T})$, it follows that $ w(\tilde{T}) \leq w(\tstar) $. This can be shown as follows.
    \begin{equation*}
    \begin{split}
        \tilde{\sigma}({\tstar})
        &= \sum_{t\in{\tstar}} \tilde{\sigma}(t)
        = \sum_{t\in{\tstar}} \left\lfloor \frac{n-2}{\varepsilon B} \cdot \sigma(t) \right\rfloor \\
        &\leq \left\lfloor \frac{n-2}{\varepsilon B} \cdot \sum_{t\in{\tstar}} \sigma(t) \right\rfloor
        \leq \left\lfloor \frac{n-2}{\varepsilon B} \cdot B \right\rfloor = \tilde{B},
    \end{split}
    \end{equation*}
as we desired.
\end{enumerate}
\end{proof}

%% file: 08_proof_dnt_special.tex
\section{Proof of \Cref{thm: dnt-special} (Special Cases for Sum-DNT)}\label{sec: dnt-special}

In this section, we provide proof for \Cref{thm: dnt-special}. We begin by restating the theorem for the reader’s convenience. Recall that $\beta = \betak$.

\ResDNTSpecial*

\begin{proof}[Proof of \Cref{thm: dnt-special} (1)]
When $ \alpha = 1 $, all output triangulations must have optimal quality.
Given $\sigma$, consider step $i$ of the algorithm in the proof of \Cref{thm: dnttheorem}(\ref{dnt: general}), and let $ S_i = \{T_1, \ldots, T_k\}$ be the set of $k$ optimal triangulations w.r.t. $\sigma$. Recall that $tr(T)$ denotes the set of triangles in the triangulation $T$.
Define the weight of a triangulation $T$ of $P$ as a 2d-vector \[ w_i(T) := \sum_{t \in tr(T)}\left(  \sigma(t), \;\; \sum_{e \in t}\sum_{T_j\in S_i}\mathbbm{1}(e\in T_j) \right), \] where the first summation is performed coordinate-wise. Observe that the first component of $w_i(T)$ represents the quality of the triangulation $T$, and the second component represents twice the total frequency of the edges of $T$ in $S_i$.
Furthermore, for any two triangulations $T'_1$ and $T'_2$ of $P$, we order $w_i(T'_1)$ and $w_i(T'_2)$ lexicographically.
Now, observe that the farthest triangulation with the optimal quality is in fact the MWT w.r.t. $w_i(\cdot)$.
Since the $k$ MWTs can be found in time $O(kn^3)$~\cite{eppstein2015k}, the overall time bound is $O(n^3k^3\log k)$.
\end{proof}

\vspace{2mm} To prove the rest, we consider two key lemmas. The first one will be used to prove \Cref{thm: dnt-special} (2). Following that, we introduce and prove an additional lemma to prove \Cref{thm: dnt-special} (3).
We begin with a lemma for finding disjoint triangulations for convex polygons when $k \leq n/2$.
\begin{restatable}{lemma}{disjtriangulations}\label{lem:disjtriangulations}
    For any $k \leq n/2$, given a convex $n$-gon, there exists an $O(kn)$-time algorithm that returns $k$ disjoint triangulations of $P$.
\end{restatable}
\begin{figure}[ht]
    \centering
    \includegraphics[width=0.8\textwidth]{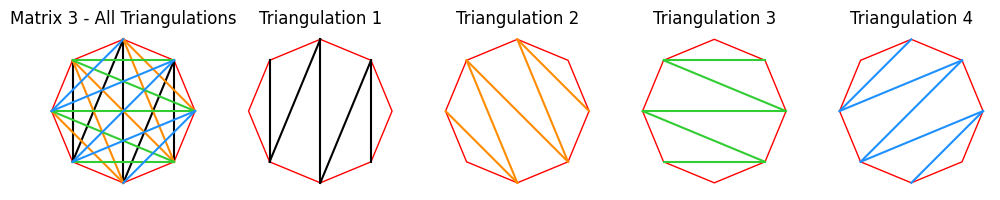}
    \caption{An example of 4 disjoint triangulations of a convex octagon.}
    \label{fig:disjointtriangulations}
\end{figure}
\begin{proof}
    We prove the case when $n$ is even. Proof of the other case is analogous. Without loss of generality, we may assume that a polygon is regular. Given $P$, fix a vertex. Call the fixed vertex $p_1$ and label the other vertices counterclockwise in increasing order. Starting from $p_1$, draw diagonals in zigzag manner as illustrated in \Cref{fig:disjointtriangulations}; $p_1 \rightarrow p_3 \rightarrow p_n \rightarrow p_4 \rightarrow p_{n-1} \rightarrow p_5 \rightarrow \cdots \rightarrow p_{3+\frac{n-4}{2}} \rightarrow p_{5+\frac{n-4}{2}} $, where the diagonals $\overline{p_1p_3}$ and $\overline{p_{3+\frac{n-4}{2}}p_ {5+\frac{n-4}{2}}}$ are symmetric, \emph{i.e.}, $\Box\,p_1p_3p_{3+\frac{n-4}{2}}p_ {5+\frac{n-4}{2}}$ is a parallelogram.
    Once done, draw another set of zigzag diagonals starting from $p_2$ in a similar manner. Repeat this process until there are no remaining diagonals. Note that, by construction, no two triangulations share a diagonal in common. Furthermore, for any triangulations we have just drawn, the first diagonal and the last diagonal are gap\text{-}1 diagonals, \emph{i.e.}, there is only one vertex between the two vertices of the diagonal. Since $ n/2 $ disjoint gap\text{-}1 diagonals can be drawn in any convex $n$-gon, we have $n/2$ disjoint triangulations.
\end{proof}

\begin{proof}[Proof of \Cref{thm: dnt-special} (2)]
    When $k \leq n/2$, we run the algorithm from \Cref{lem:disjtriangulations}, which takes $O(kn)$ time.
    When $k \geq 2/\varepsilon$, run the $ O(n^5k^5\log k) $-time algorithm illustrated in \Cref{thm: dnttheorem}.
    Finally, when $n/2 < k < 2/\varepsilon$, we exhaustively check all sets of $k$ triangulations of $P$. Since a convex $r$-gon has exactly the $(r-2)$-nd Catalan number of triangulations and the $r$-th Catalan number is $O(4^{r})$~\cite{DUTTON1986211}, this step can be done in time $2^{O(1/\varepsilon^2)}$.
    Combining the results from all three cases completes the proof of the desired time bound.
\end{proof}

Observe that a polygon admits a unique Delaunay triangulation unless it has a set of four or more co-circular points~\cite{deberg2008}.
Consequently, to obtain \emph{diverse} Delaunay triangulations, it suffices to focus on the sets of co-circular points.
Any such co-circular set forms a convex polygon, each co-circular set with $r$ points has $O(4^r)$ possible triangulations. Therefore, a naive approach for finding $k$ triangulations with maximum diversity in a polygon that has $m$ co-circular sets may require 
\[
O\left( \binom{4^{r_1} \times 4^{r_2} \times \dots \times 4^{r_m}}{k} \cdot nk^2 \right)
=\,
O\bigl(4^{k(r_1 + \dots + r_m)} \cdot nk^2\bigr)
=\,
2^{O(kn)},
\]
where $r_1 + \dots + r_m = O(n)$. This exponential bound indicates the need for a more efficient technique to achieve a PTAS. To this end, we introduce the following lemma.

\begin{lemma}\label{lem: delaunay-brute-force}
Let $P$ be a Delaunay-triangulable polygon with $n$ vertices. Let $C_1$, ..., $C_m$ be the sets of co-circular points in $P$, each corresponding to a distinct co-circle and containing at least four points, and let $ M = \max_{\ell \in [m]} \lvert C_\ell \rvert. $
Then there is an algorithm running in $m^{O(k^2)} 2^{O(kM)}$ time that computes $k$ distinct Delaunay triangulations $T_1, \dots, T_k$ of $P$ with \emph{maximum diversity}, i.e., for any other set of $k$ distinct Delaunay triangulations $T'_1, \dots, T'_k$ of $P$, $ \sum_{i \neq j} \lvert T_i \Delta T_j \rvert 
\geq
\sum_{i \neq j} \lvert T'_i \Delta T'_j \rvert. $
\end{lemma}

\begin{proof}
Outside of the co-circular sets $C_1,\dots,C_m$, $P$ has a unique Delaunay triangulation (denote it by $T'$). Let $T^{C_\ell}$ be a triangulation of $C_\ell$. Then any Delaunay triangulation $T$ of $P$ can be represented as
\[
T \;=\; T^{C_1} \;\cup\; T^{C_2} \;\cup\;\cdots\;\cup\; T^{C_m} \;\cup\; T',
\]
where $T^{C_\ell}$ is chosen from among at most $O(4^M)$ triangulations of $C_\ell$.

Two triangulations $T_i$ and $T_j$ are distinct if and only if there is some $\ell \in [m]$ for which $T^{C_{\ell}}_i \neq T^{C_{\ell}}_j$. Furthermore, for any collection of $k$ triangulations $T_1, \dots, T_k$, there are $m^{\binom{k}{2}}$ ways to choose a triple $(i,j,\ell)$ that indicates which co-circular set distinguishes $T_i$ from $T_j$.

Suppose we fix an $\ell_{ij}\in [m]$ for each pair $(i,j)$. To ensure $T_i \neq T_j$, assign two distinct triangulations of $C_{\ell_{ij}}$ to $T_i$ and $T_j$ among the $O(4^M)$ possibilities, and for any other $T^{C_{\ell_{ij}}}_{i'}$ (where $i'\in[k]$ is not in the pair), assign an arbitrary triangulation of $C_{\ell_{ij}}$. If this process yields a valid set of $k$ distinct triangulations, we compute its diversity; otherwise, we assign diversity $-\infty$.

As each $C_\ell$ has at most $O(4^M)$ possible triangulations, there are
\[
O\bigl(m^{\binom{k}{2}} \cdot (4^M)^k \cdot m\bigr)
\]
such assignments. We select the configuration with the greatest diversity and return its associated triangulations. Implementing this procedure takes 
\[
O\bigl(m^{\binom{k}{2}} \cdot (4^M)^k \cdot m\bigr) \cdot k^2n 
\;=\;
m^{O(k^2)} 2^{O(kM)},
\]
as required.
\end{proof}

\vspace{2mm}We are now ready to prove \Cref{thm: dnt-special}(3).

\begin{proof}[Proof of \Cref{thm: dnt-special} (3)]
Recall that determining whether $P$ admits a Delaunay triangulation can be done in $O(n \log n)$ time, and a simple $n$-gon has multiple Delaunay triangulations if and only if it contains sets of at least four co-circular points~\cite{deberg2008}.
Such co-circular points can be identified in $O(n \log n)$ time by constructing the Voronoi diagram and scanning for vertices of degree at least four.
Since each co-circular set forms a convex polygon, and a convex $r$-gon has $(r-2)$-nd Catalan number many triangulations, we can also check in $O(n \log n)$ time whether $P$ has at least $k$ distinct Delaunay triangulations.

Assume now that $P$ has $m$ co-circular sets $C_1,\ldots,C_m$. As described in \Cref{lem: delaunay-brute-force}, every Delaunay triangulation $T$ of $P$ can be represented as
\[
T \;=\; T^{C_1} \;\cup\; T^{C_2} \;\cup\;\cdots\;\cup\; T^{C_m} \;\cup\; T',
\]
where $T^{C_\ell}$ is a triangulation of the convex polygon $C_\ell$, and $T'$ is the unique triangulation of the remaining points (those not in any $C_\ell$). Let $ M = \max_{\ell \in [m]} \lvert C_{\ell} \rvert. $

When $k>\varepsilon/2$, we may use the algorithm in (1).
We therefore focus on the remaining case $k < 2 / \varepsilon$, which we divide into two subcases:

\vspace{2mm}\textbf{(Case 1: $k \le |C_{\ell}| / 2$ for some $\ell\in[m]$)}
Using \Cref{lem:disjtriangulations}, we can find $k$ \emph{disjoint} Delaunay triangulations $T_1^{C_{\ell'}}, \ldots, T_k^{C_{\ell'}}$ of every $C_{\ell'}$ with $|C_{\ell'}| \ge |C_{\ell}|$, in total $O\bigl(|C_{\ell'}|k\bigr)$ time per set.  Since $\sum_{\ell'} |C_{\ell'}| \le n$, this step takes $O(kn)$ time overall.

For each $C_{\ell''}$ such that $|C_{\ell''}|/2 < k$, we instead perform a brute-force search to obtain $k$ (possibly repeated) \emph{diverse} Delaunay triangulations $T_1^{C_{\ell''}}, \ldots, T_k^{C_{\ell''}}$.  Because $|C_{\ell''}|/2 \le M/2 < k$, and $k < 2/\varepsilon$, this takes at most $ O( {4^k \choose k} \cdot n) = 2^{O(\varepsilon^{-2})} \cdot n $
using the fact that $\binom{4^k}{k}$ is $2^{O(k^2)}$ and $k = O(\varepsilon^{-1})$.

\vspace{2mm}\textbf{(Case 2: $|C_{\ell}|/2 < k$ for every $\ell\in[m]$)}
Here, we apply \Cref{lem: delaunay-brute-force} directly to each $C_{\ell}$, and find $k$ disjoint triangulations $T_1^{C_{\ell}}, \ldots, T_k^{C_{\ell}}$ of every $C_{\ell}$. Since $|C_{\ell}|/2 \le M/2 < k < 2/\varepsilon$, the number of ways to pick triangulations is at most $2^{O(kM)}$, and combining them for the $m = O(n)$ sets yields $n^{O(\varepsilon^{-2})}$ time in total.

\vspace{2mm}
In each case, we construct 
\[
T_i \;=\; T_i^{C_1} \;\cup\; T_i^{C_2} \;\cup\;\cdots\;\cup\; T_i^{C_m} \;\cup\; T',
\]
for $i\in[m]$, and return $T_1,\dots,T_k$.
The overall running time across the two subcases as well as the case $ k \geq 2/\varepsilon $ is bounded by $ n^{O(\varepsilon^{-2})}+O(n^3k^3\log k) $ as claimed. This completes the proof.
\end{proof}

%% file: 09_aqnd.tex
\section{Proof of \Cref{thm:codes} (Reduction to Hamming Codes)}\label{sec: aqnd-proof}

In this section, we provide proof of \Cref{thm:codes}. We restate the theorem here for reader's convenience.
\cctrihardness*
As mentioned earlier, computing $A_2(n,d)$ is still open, and only a limited number of instances are currently known, e.g., see~\cite{ostergard2011size}. 
\begin{proof}[Proof of \Cref{thm:codes}]
Consider the recursively-defined simple polygon $P$ as illustrated in \Cref{fig:spiral}. Every triangulation of $P$ must contain the diagonals with color red. Each of the remaining non-triangulated regions is a convex quadrilateral, which can be triangulated in two different ways. Note that the choices of triangulations in the convex quadrilaterals can be made independently. Given an integer $n$, we can let the recursion repeat until it contains $n$ such convex quadrilaterals. We call this simple polygon $P_n$.
\begin{figure}[H]
\centering
\includegraphics[scale=0.12]{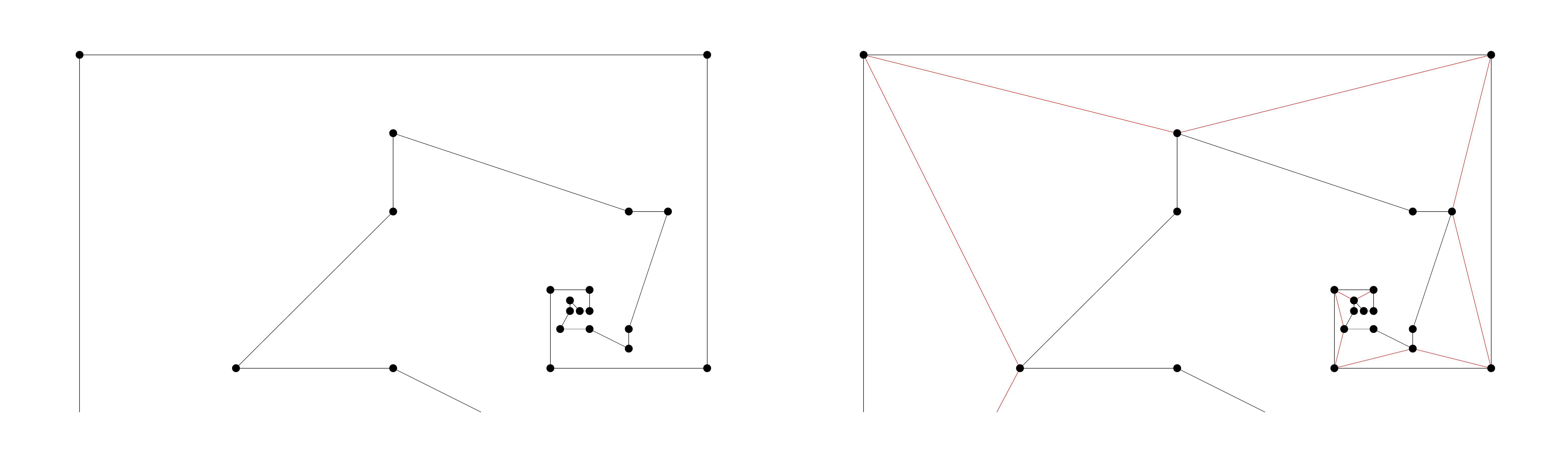}
\caption{A recursively-defined simple polygon $P$.\label{fig:spiral}}
\end{figure}
Our goal is to perform a reduction from computing $A_2(n, d)$ to finding a diverse set of $k = O(n)$ triangulations for $P_n$. Since we require $d > n/2$, $A_2(n, d) = O(n)$, as mentioned at the beginning of this section. The reduction works as follows. We perform a binary search on $g$ in the range $[1, O(n)]$. To verify whether $A_2(n, d) \ge g$, we can ask whether $P_n$ contains $k = g$ triangulations with minimum pairwise Hamming distance at least $\delta = d$. If the answer is ``Yes,'' set $g$ to be a larger value; otherwise, set $g$ to be a smaller value. Thus, we can compute $A_2(n, d)$ by invoking the triangulation problem $O(\log n)$ times. 
\end{proof}

%% file: 10_Max_Min.tex
\section{Proof of \Cref{thm: mindttheorem} (Algorithms for Min-DT)}
\label{sec: proof-min-dt}

In this section, we prove \Cref{thm: mindttheorem}. We begin by restating the theorem.

\mindttheorem*

Recall that the simple farthest insertion algorithm w.r.t. $\MinSD$ also gives $\tfrac12$-approximate solutions\cite{ravi1994heuristic}.
To implement farthest insertion w.r.t. $\MinSD$, we define Multi-criteria Triangulation (MCT) problem, a natural generalization of the BCT problem (\Cref{def: bct}). 
We then use MCT to design an FPT algorithm for the decision version of the farthest insertion w.r.t. $\MinSD$ measure.

\begin{definition}[Multi-Criteria Triangulation] Given an objective function $ \sigma: \cT \rightarrow \mathbbm{R} $ and a collection of constraints $\mathcal{B}$ on $\cT$, the \emph{Multi-Criteria Triangulation} problem $MCT(\sigma,\mathcal{B})$ is defined as follows:
\begin{align}
    \text{Minimize} \quad & \sigma(T) \quad\text{over all $T\in\cT{}$}\\
    \text{subject to} \quad & \text{bool}(T) = \textsc{true} \quad\text{for all $\text{bool}\in\mathcal{B}$}.
\end{align}
Its maximization counterpart is defined analogously.
\end{definition}\label{def: mct}

The following lemma is also a general extension to BCT theorem (\Cref{thm: bct-algorithms}) and can be proved similarly. We omit the details here.
\begin{lemma}\label{lem: mct-integer}
    Assume that $MCT(\sigma, \{\sigma_1 \leq b_1, 
    \ldots, \sigma_i \leq b_i\} )$ is given such that $\sigma_j: \cT \rightarrow \mathbbm{Z}_{\geq0} $ and $b_j \in \mathbbm{Z}_{\geq0}$, where $j\in[i]$.
    Then, $MCT(\sigma, \{\sigma_1 \leq b_1, 
    \ldots, \sigma_i \leq b_i\} )$ can be solved in time $O\left((b^2_1 + 1)\cdots (b^2_i + 1) \cdot n^3\right)$.
\end{lemma}

We now present an FPT algorithm for the decision version of the farthest insertion under $\MinSD$ measure.

\begin{restatable}[Decision Version of Farthest Insertion Under $\MinSD$]{theorem}{ResMMFTFPT}\label{thm: MMFT-fpt}
    Given polygon $P$ and a set $ \cT_k = \{T_1,\ldots,T_k\}$ of $k$ triangulations of $P$, there is an $ O(k\cdot{}(r+1)^{2k-2}\cdot{}n^3) $-time algorithm that finds a triangulation $T$ such that $ \min_{j} |T \Delta T_j| \geq 2(n-3)-2r $.
\end{restatable}

\begin{proof}
    Recall from \Cref{prop: min-ce} that minimizing the number of common edges between two triangulations is equivalent to maximizing the symmetric difference between them.
    Our goal is to find a triangulation $T$ such that $|T_j\Delta T| \leq r $ for every $ T_j\in \cT_k $.
    
    For $j\in[k]$, let $ \sigma_j (e) := \mathbbm{1}(e \in T_j) $ and let $ \sigma_j(T) := \sum_{e\in{}T} \mathbbm{1}(e\in{}T_j) $.
    In other words, $ \sigma_j(T) $ denotes the number of common edges between $T$ and $T_j$.
    For a fixed $i$, consider $ MCT(\sigma_i, \{\sigma_j(T)\leq{}r\}_{j\in[k]-i}) $, which solves the following problem in $O\left(\left(r+1\right)^{2(k-1)}\cdot{}n^3\right)$ by \Cref{lem: mct-integer}:
        \begin{align}
            \text{Minimize} \quad & \sigma_i(T) \quad\text{over all $T\in\cT{}$}\\
            \text{subject to} \quad & \sigma_j(T) \leq r\quad\text{for all $j\in{}[k]-i$,}
        \end{align}
    Note that if $ \sigma_i(T)\leq{}r $ and some $T\in\cT$, where all the constraints are satisfied, then $ |T \Delta T_j| \geq 2(n-3)-2r $ for all $j\in[k]$. Therefore, we have the following simple algorithm:
        \begin{enumerate}
            \item Set $i\leftarrow{}1$.
            \item While $ i \leq k $ \label{alg: multi-bct-max-while}
            \begin{enumerate}
                \item Let $T\leftarrow{}MCT(\sigma_i(T), \{\sigma_j(T)\leq{}r\}_{j\in[k]-i})$.
                \item If $ \sigma_i(T) \leq r $, return $T$.
                \item If not, increase $i$ by 1 and go to Step~\ref{alg: multi-bct-max-while}.
            \end{enumerate}
        \item Return $\bot$.
        \end{enumerate}
    Since the second step repeats at most $k$ times, we have the desired running time.
\end{proof}

Now, we provide proof for \Cref{thm: mindttheorem}.
\begin{proof}[Proof of \Cref{thm: mindttheorem}]
We begin with an empty collection, and then we update this collection incrementally by adding the triangulation obtained by farthest insertion algorithm w.r.t. $\MinSD$.

By \Cref{thm: MMFT-fpt}, given a collection of triangulations $ \cT_i = \{T_1,\ldots,T_i\}$ and $r\in[0,n-3]$, one can find a triangulation $T$ such that $ \min_{j\in[i]}|T_j\Delta T| \geq 2(n-3)-2r $, if it exists.
Thus, by running this algorithm for every $r\in[0,n-3]$ starting from 0, we may find the farthest triangulation from $\cT_i$ w.r.t. $\MinSD$.
Note that once the farthest triangulation is found for some $r$, we do not have to check further since as $r$ increase the $\MinSD$ measure only decreases.

Let $T_{i+1}$ be the farthest triangulation from $\cT_i$, and let $ d_{\mathrm{OPT}_{i+1}} $ denote the minimum pairwise distance between $T_{i+1}$ and triangulations in $\cT_i$.
I.e.,
\[
T_{i+1} = \argmax_{T\in\cT}\min_{j\in[i]}|T_j\Delta T| \qquad \text{and} \qquad d_{\mathrm{OPT}_{i+1}} = \min_{j\in[i]}|T_j\Delta T_{i+1}|.
\]
Let $r_{i+1}$ be the smallest value of $r\in[0,n-3]$ such that $ d_{\mathrm{OPT}_{i+1}} = 2(n-3)-2r_{i+1} $.
Then, one can find $T_{i+1}$ in time $ \sum_{r=0}^{r_{i+1}} O\left(i\cdot (r_{i+1}+1)^{2k-2} \cdot n^3 \right) $, which is simply $ O\left( i\cdot (r_{i+1}+1)^{2k-1} \cdot n^3 \right) $.
Since finding $T_k$ takes longer than finding $T_1,\ldots,T_{k-1}$, the overall running time of this algorithm will be
\[
O\left( k^2\cdot (r_k+1)^{2k-1} \cdot n^3 \right),
\]
which is simply $ r_k^{O(k)} $.
Since $ d_{\mathrm{OPT}_k} \geq \frac{1}{2}\cdot d_{\mathrm{OPT}} $, we have the desired result.
\end{proof}